\documentclass[11pt,reqno]{amsart}
\usepackage{amsmath,amssymb, mathrsfs, amsfonts}
\usepackage{graphicx, color, amscd}
\usepackage{comment}
\usepackage{slashed}
\newcommand{\RN}[1]{%
  \textup{\uppercase\expandafter{\romannumeral#1}}%
}
\newcommand{\pl}{\partial}

\newcommand{\rw}{\rightarrow}

\newcommand{\R}{\mathbb{R}}

\newcommand{\lt}{\left}
\newcommand{\rt}{\right}
\newcommand{\na}{\nabla}
\renewcommand{\hat}{\widehat}
\newcommand{\cb}{\underline{c}}

\newcommand{\bna}{\bar\nabla}

\newtheorem{theorem}{Theorem}
\newtheorem{proposition}[theorem]{Proposition}

\newtheorem{corollary}[theorem]{Corollary}
\newtheorem{lemma}[theorem]{Lemma}

\theoremstyle{definition}
\newtheorem{definition}[theorem]{Definition}
\newtheorem{remark}[theorem]{Remark}

\title{BMS charges without supertranslation ambiguity}
\author[P.-N. Chen, M.-T. Wang, Y.-K. Wang, and S.-T. Yau]{Po-Ning Chen, Mu-Tao Wang,\\ Ye-Kai Wang, and Shing-Tung Yau}
\numberwithin{theorem}{section}
\numberwithin{equation}{section}

\begin{document}

\begin{abstract} 
The asymptotic symmetry of an isolated gravitating system, or the Bondi-Metzner-Sachs (BMS) group, contains an infinite-dimensional subgroup of supertranslations.  
Despite decades of study, the difficulties with the  ``supertranslation ambiguity" persisted in  making sense of fundamental notions such as the angular momentum carried away by gravitational radiation. The issues of angular momentum and center of mass were resolved by the authors  recently. In this paper, we address the issues for conserved charges with respect to both the classical BMS algebra and the extended BMS algebra.  In particular, supertranslation ambiguity of the classical charge for the BMS algebra, as well as the extended BMS algebra, is completely identified. We then propose a new invariant charge by adding correction terms to the classical charge. With the presence of these correction terms, the new invariant charge is then shown to be free from any supertranslation ambiguity. Finally, we prove that both the classical and invariant charges
for the extended BMS algebra are invariant under the boost transformations.
\end{abstract}

\thanks{P.-N. Chen is supported by NSF grant DMS-1308164 and Simons Foundation collaboration grant \#584785, M.-T. Wang is supported by NSF grants DMS-1810856 and DMS-2104212, Y.-K. Wang is supported by MOST Taiwan grant 107-2115-M-006-001-MY2 and 109-2628-M-006-001 -MY3, and S.-T. Yau is supported by NSF grants  PHY-0714648 and DMS-1308244.} 

\maketitle
\tableofcontents
\section{Introduction}

The study of gravitational radiation \cite{Bondi, BVM} in the 1960's culminated in the discovery of the definition of Bondi-Sachs energy-momentum and the mass loss formula of Bondi. An unexpected realization was that the asymptotic symmetry group of an asymptotically flat spacetime, or the Bondi-Metzner-Sachs (BMS) group \cite{BVM, Sachs, Sachs2},  is much larger than the Poincar\'e group. In particular, there exists an infinite dimensional normal subgroup of supertranslations of the BMS group.  This was expressed first in terms of coordinate transformations of Bondi-Sach coordinate systems \cite{BVM, Sachs} and then in terms of Penrose's conformal compactifications \cite{NP2, Geroch}. Supertranslations are universal as they are independent of the underlying spacetime, and ubiquitous in any description of null infinity. Their presence created subtleties in describing physical reality for a distant observer. The Bond-Sachs energy-momentum was shown to be supertranslation invariant and the energy/mass loss due to radiation can be evaluated without ambiguity. However, this is no longer the case for angular momentum at null infinity. With the existence of supertranslation ambiguity, one cannot rigorously discuss the angular momentum carried away by gravitational radiation \cite{Penrose2}. This issue was recently resolved by the authors' discovery \cite{CKWWY, CWWY} of a supertranslation invariant definition of angular momentum, which modifies the classical definition of angular momentum by a correction term that is derived from the CWY quasilocal angular momentum \cite{CWY3, CWY4}.  A supertranslation invariant
center of mass was also proposed in \cite{CKWWY, CWWY}. Angular momentum and center of mass correspond to the quotient Lorentz algebra (rotations and boosts) of the BMS algebra by supertranslations. In this article, we consider the conserved charges \cite{AS, DS, WZ} of the full BMS algebra. Our analysis also applies to the extended BMS algebra which includes superrotations or superboosts \cite{Flanagan-Nichols}. In both cases, supertranslation ambiguity is completely identified for the classical charge, and an invariant charge is proposed and shown to be free from any supertranslation ambiguity. 
We note that, from a completely different perspective, the work of Hawking-Perry-Strominger \cite{HPS} attempted to interpret these ambiguities associated with superrotations and supertranslations as soft hairs of black holes. 

Consider the future null infinity $\mathscr{I}^+$  of an asymptotically flat spacetime which is described in terms of a Bondi-Sachs coordinate system. $\mathscr{I}^+$ is identified with $I \times S^2$, where $I\subset (-\infty, +\infty)$ is an interval parametrized by the retarded time $u$ and $S^2$ is the standard unit 2-sphere equipped with the standard round metric $\sigma_{AB}$.  
Let $\epsilon_{AB}$ be the area form of $\sigma_{AB}$ and $\na_A$ be the covariant derivative with respect to $\sigma_{AB}$. Indexes are raised, lowered, or contracted with respect to $\sigma_{AB}$. 

The symmetry of $\mathscr{I}^+$ consists of the Bondi-Metzner-Sachs (BMS) group, of which supertranslations form an infinite dimensional subgroup. 
\begin{definition} Suppose $(\bar{u}, \bar{x})$ and $(u, x)$ are two Bondi-Sachs coordinate systems on $\mathscr{I}^+$. 
A BMS transformation $B: \mathscr{I}^+ \rw \mathscr{I}^+$ is given by
\begin{align}\label{BMS_transformation}
(u, x) = B(\bar u, \bar x) = \lt( K(\bar x) \bar u + f(\bar x), g(\bar x) \rt)
\end{align}
where $f$ is a function on $S^2$ and $g: (S^2, \bar{\sigma}) \rw (S^2, \sigma)$ is a conformal map with $g^*\sigma = K^2 \bar \sigma$ (see Appendix C).

$B$ is called a {\it supertranslation}, denoted by $S_f$, if $g$ is the identity map, $K\equiv 1$, and $(u, x)=(\bar{u}+f(\bar{x}), \bar{x})$. On the other hand, $B$ is called a {\it boost}, denoted by $B_g$,  if $f=0$ 
and $(u, x)=(K(\bar{x})\bar{u}, g(\bar x))$. 
\end{definition}

The infinitesimal symmetry of $\mathscr{I}^+$ is the BMS algebra of BMS fields. Our consideration also includes the extended BMS algebra 
which was recently proposed in \cite{Banks, BT, BT1, BT2}. 

\begin{definition}
A {\it classical BMS field} is a vector field $Y$ on $\mathscr{I}^+$
\begin{equation}\label{BMS_field} Y=(Y^0+uY^1)\partial_u+Y^A \partial_A, \end{equation} such that $Y^0$ is a smooth function on $S^2$ and $Y^A$ is a smooth conformal Killing vector field on $S^2$,
\begin{equation}\label{cke_intro} \nabla^A Y^B+\nabla^B Y^A=2Y^1\sigma_{AB}.\end{equation}

An {\it extended BMS field} is a vector field $Y$ defined on $(-\infty,\infty) \times\mathscr{U}\subset \mathscr{I}^+$ for some open subset $\mathscr{U}\subset S^2$. $Y = (Y^0 +uY^1)\partial_u+Y^A \partial_A$  where $Y^0$ is a smooth function on $\mathscr{U}$ and $Y^A$ satisfies \eqref{cke_intro} on $\mathscr{U}$.  It is understood that an extended BMS field indicates such a pair $(Y, \mathscr{U})$, though $\mathscr{U}$ is often omitted. 
\end{definition}

A classical BMS field corresponds to an infinitesimal deformation of the BMS group, while an extended BMS field may not be integrable. We also note that a classical BMS field is characterized by that both $\na_A Y^A$ and $\epsilon_{AB}\na^A Y^B$ are functions of harmonic mode $\ell=1$ on $S^2$. One can view a classical field as a special case of an extended field with $\mathscr{U} = S^2$.

Let $m$ denote the mass aspect, $N_A$ the angular momentum aspect, $C_{AB}$ the shear tensor, and $N_{AB}$ the news tensor of the Bondi-Sachs coordinate on $\mathscr{I}^+$.  One can view $m$ as a smooth function, $N_A$ a smooth one-form, and $C_{AB}$ and $N_{AB}$ smooth symmetric traceless 2-tensors (with respect to $\sigma_{AB}$) on $S^2$ that depend on $u$. In particular, $\partial_u C_{AB}=N_{AB}$. See a brief description of $\mathscr{I}^+$ in the Bondi-Sachs coordinate and the definitions of these quantities in Section 2.1.

\begin{definition}
The {\it classical charge} $\tilde{Q}(u, Y)$ associated with an extended BMS field $Y$ defined on $(-\infty, \infty)\times \mathscr{U}$ is defined to be:
\begin{align}\label{classical}
\begin{split}
\tilde Q(u, Y)  &= \int_{\mathscr{U} } Y^A [N_A-\frac{1}{4}C_{A}^{\,\,\,\,D}\nabla^B C_{DB} - \frac{1}{16}\nabla_{A}(C_{DE}C^{DE})] \\
&\quad +\int_{\mathscr{U}} (Y^0+uY^1) (2m)\\
&\quad -\frac{1}{4} \int_\mathscr{U} u (\epsilon_{AB} \nabla^A Y^B) \epsilon^{PQ} \nabla_P \nabla^E C_{EQ} \end{split}                  
\end{align} where $\nabla_A$ denotes the covariant derivative with respect to $\sigma_{AB}$. 
\end{definition}

The last integral vanishes for classical BMS fields as $\epsilon_{AB}\na^A Y^B$ is of mode $\ell=1$. In this case, the definition reduces to that of Dray-Streubel \cite{DS}. See Section III.B of Flanagan-Nichols \cite{Flanagan-Nichols} for details. It is added to make the charge independent of retarded time in a non-radiative spacetime for extended BMS fields. See also Section III.D of Flanagan-Nichols \cite{Flanagan-Nichols} for a related discussion.  The integrals may not be finite for an extended BMS field $Y$. In this case, we shall assume that the supports of the BMS functions $m, N_A, C_{AB}, N_{AB}$ are contained in $\mathscr{U}$. This will ensure the finiteness of the charge and justify the integration by parts that we perform. One can in fact work in the space of the $L^2$ completion of data with such support. For simplicity, we will not pursue this completion here. In the Bondi-Sachs formalism, such data can be constructed by choosing initially that $m$, $C_{AB}$ and $N_A$ are properly supported and then require that $N_{AB}$ is properly supported. On the other hand, asymptotically flat initial data sets at spacetime infinity with physical data properly supported can be constructed via the localized solutions to the Einstein constraints equation by Carlotto and Schoen \cite{CS}.

Throughout this article, we assume that $\mathscr{I}^+$ extends from $u=-\infty$ ($i^0$) to $u=+\infty$ ($i^+$), the total flux of the classical charge is \[\delta \tilde{Q}(Y)=\lim_{u\rightarrow +\infty} \tilde{Q}(u, Y)-\lim_{u\rightarrow +\infty} \tilde{Q}(u, Y).\]

We turn to the definition of supertranslation ambiguity.
\begin{definition} Let $S_f$ be the supertranslation $(u, x)=(\bar{u}+f(\bar{x}), \bar{x})$. 
We say that two extended BMS fields $Y=(Y^0+uY^1)\partial_u+Y^A \partial_A$ and $\bar{Y}=(\bar{Y}^0+\bar{u}\bar{Y}^1)\partial_{\bar{u}}+\bar{Y}^a \partial_{a}$ are related by $S_f$ if 
\begin{equation} \label{supertranslation_related}Y^0=\bar{Y}^0 , Y^1=\bar{Y}^1, Y^A=\delta^A_a\bar{Y}^a.\end{equation}
\end{definition}
 
\begin{definition}
The {\it supertranslation ambiguity} of the total flux of the classical charge is defined as 
\[\delta \tilde{Q}(\bar{Y})-\delta \tilde{Q}(Y)\]
for two BMS fields related by a supertranslation.
\end{definition}

\begin{definition} We say a quantity is {\it supertranslation invariant} if it is invariant under all supertranslations of mode $\ell\geq 2$. In the case of the classical charge,
this means \[\delta \tilde{Q}(\bar{Y})-\delta \tilde{Q}(Y)=0\] for any supertranslation $S_f$ of mode $\ell\geq 2$. 

\end{definition}

For classical BMS fields, we prove that

\begin{theorem}[Corollary \ref{supertranslation invariance of total flux}] Suppose the news tensor decays as
\[ N_{AB}( u,x) = O(|u|^{-1-\varepsilon}) \mbox{ as } u \rw \pm\infty,
\]
then the total flux of the classical charge is supertranslation invariant for all classical BMS fields if and only if \[\lim_{u \rw +\infty} m(u,x) - \lim_{u \rw -\infty}m(u,x)\] is a constant on $S^2$. 
\end{theorem}
In particular, as long as $\lim_{u \rw +\infty} m(u,x) - \lim_{u \rw -\infty}m(u,x)$ is non-constant, the total flux of the classical charge is shifted by supertranslations and can assume any value. The supertranslation ambiguity of an extended BMS field is also identified, see Theorem \ref{extended_flux_ambiguity}.

In order to remove the supertranslation ambiguity, we define a new charge for the BMS algebra by adding correction terms. To define the new charge, we consider the decomposition of $C_{AB}$ into 
\begin{equation}\label{CAB}C_{AB}=\nabla_A\nabla_B c-\frac{1}{2} \sigma_{AB} \Delta c+\frac{1}{2}(\epsilon_A^{\,\,\,\, E} \nabla_E \nabla_B \underline{c}+\epsilon_B^{\,\,\,\, E} \nabla_E \nabla_A \underline{c})\end{equation} where $\epsilon_{AB}$ denotes  the area form of $\sigma_{AB}$.  $c=c(u, x)$ and $\underline{c}=\underline{c}(u, x) $ are the closed and co-closed potentials of $C_{AB}(u, x)$. They are chosen to be of $\ell\geq 2$ harmonic modes and thus such a decomposition is unique. 

Given the decomposition of the shear tensor $C_{AB}$, its Hodge dual, $(*C)_{AB}={\epsilon_A}^DC_{DB}$, admits the following decomposition 
\[
(*C)_{AB}=-\nabla_A\nabla_B \underline{c} + \frac{1}{2} \sigma_{AB} \Delta \underline{c}+\frac{1}{2}(\epsilon_A^{\,\,\,\, E} \nabla_E \nabla_B c+\epsilon_B^{\,\,\,\, E} \nabla_E \nabla_A c).
\]
In the following, we propose a notion of invariant charge for an extended BMS field.
\begin{definition}\label{def_invariant_charge} The {\it invariant charge} $Q(u, Y)$ associated with an extended BMS field $Y$ on $(-\infty, \infty)\times \mathscr{U}$ is defined to be:
\begin{align}\label{invariant_charge} \begin{split}
  Q(u, Y) &= \tilde Q(u,Y) + \int_\mathscr{U} (Y^A \na_A c - Y^1 c )m\\
&\quad - \frac{1}{16}\int_\mathscr{U} c \na^A\na^B(\na_D Y^D)C_{AB}\\
&\quad -\frac{1}{16}\int_\mathscr{U}  \underline{c} \na^A\na^B(\na_D Y^D)  (*C)_{AB} .
\end{split}  \end{align} 
\end{definition}
Again, we assume that $m$, $C_{AB}$ and $N_A$ are supported on $\mathscr{U}$.  

\begin{remark}
For a classical BMS field $Y$, 
\begin{align}\label{invariant_charge_classical}
\begin{split}
Q(u,Y) &= \int_{S^2 } Y^A [N_A-\frac{1}{4}C_{A}^{\,\,\,\,D}\nabla^B C_{DB} - \frac{1}{16}\nabla_{A}(C_{DE}C^{DE})] \\
&\quad +\int_{S^2} (Y^0+uY^1) (2m) + \int_{S^2} \lt( Y^A \na_A c - Y^1 c\rt)m.
\end{split}
\end{align}
This provides a unified expression for the invariant angular momentum ($Y^A = \epsilon^{AB}\na_B \tilde X^k$) and invariant center of mass integral ($Y^A = \na^A \tilde X^k$) studied in \cite{CWWY}. We do not include the two terms involving co-closed potentials in Chen-Wang-Yau center of mass integral in order to retain boost invariance.
\end{remark}
The correction term for a classical BMS field is exactly the same as the correction term for the invariant angular momentum and center of mass defined in \cite{CKWWY, CWWY}, which are the limits of CWY quasilocal angular momentum and center of mass at null infinity \cite{CWY3, CWY4, KWY}. The correction term arises from solving the optimal isometric embedding equation in the theory of Wang-Yau quasilocal mass \cite{WY1, WY2}. This provides the reference term  that is critical in the Hamiltonian approach of defining conserved quantities.

We can similarly define the total flux of our new invariant charge and the supertranslation ambiguity. We prove that the total flux of the invariant charge is free of supertranslation ambiguity for both classical BMS fields and extended BMS fields.

\begin{theorem}[Theorem \ref{classical_invariant}]
Suppose $Y$ is a classical BMS field and $N_{AB}=O(|u|^{-1-\epsilon})$ as $u\rightarrow \pm \infty$, then the total flux of the invariant charge ${Q}(u, Y)$ is  supertranslation invariant.
\end{theorem}

In order to ensure the finiteness of the total flux and prove the supertranslation invariance for an extended BMS field, a slightly stronger decay condition for the news at $i^0$ and $i^+$ is required. 

\begin{theorem}[Theorem \ref{extended_invariant}]
Suppose $Y$ is an extended BMS field and $N_{AB}=O(|u|^{-2-\epsilon})$ as $u\rightarrow \pm \infty$, then the total flux of the invariant charge ${Q}(u, Y)$ is supertranslation invariant.
\end{theorem}

Next, we study the invariance of the total flux of the charges under a boost, first for the classical BMS fields and next for the extended BMS fields. 

\begin{definition} \label{def_boost_related} Let $B_g$ be the boost $(u, x)=(K(\bar{x}) \bar{u}, g(\bar{x}))$. 
We say that two extended BMS fields  $Y=(Y^0+uY^1)\partial_u+Y^A \partial_A$ and $\bar{Y}=(\bar{Y}^0+\bar{u}\bar{Y}^1)\partial_{\bar{u}}+\bar{Y}^a \partial_{a}$ are \it{related by $B_g$} if 
\begin{equation}Y^0=K\bar{Y}^0, Y^1=\bar{Y}^1+ K^{-1}\bar Y^a \na_a K, Y^A=\partial_a g^A \bar{Y}^a\end{equation}
where the left hand sides are evaluated at $x$ and  the right hand sides are evaluated at $\bar x = g^{-1}(x)$.
\end{definition}

\begin{theorem}[Theorem \ref{boost_invariance_flux}]
Suppose $\bar Y$ and $Y$ are two classical BMS fields related by a boost and $N_{AB}=O(|u|^{-1-\epsilon})$ as $u\rightarrow \pm \infty$. We have
\[
\begin{split}
\delta \tilde Q(\bar Y) = \delta \tilde Q (Y) \text{ and }
\delta  Q(\bar Y) = \delta Q (Y).
\end{split}
\]
\end{theorem}

\begin{theorem}[Theorem \ref{boost_equ_extended_flux}]
Suppose $\bar Y$ and $Y$ are two extended BMS fields related by a boost and $N_{AB}=O(|u|^{-2-\epsilon})$ as $u\rightarrow \pm \infty$. We have
\[
\begin{split}
\delta \tilde Q(\bar Y) = \delta \tilde Q (Y) \text{ and }
\delta  Q(\bar Y) =  \delta Q (Y).
\end{split}
\]
\end{theorem}

The other major results of this article concern non-radiative spacetimes\footnote{A spacetime is non-radiative if the news tensor vanishes.}. In such case, $Q(u, Y)$ is supertranslation invariant and boost invariant. See the statement of Theorem \ref{invariant_CWY_vanish_news}, Theorem \ref{boost_invariance_nonradiative}, Theorem \ref{invariant_CWY_vanish_news_extended} and Theorem \ref{boost_equ_extended_vanish} for further details.

The paper is organized as follows. In Section 2, we review the Bondi-Sachs coordinate system, compute the evolution of classical charges, and present the formula of Bondi-Sachs data in non-radiative spacetimes. In Section 3, the supertranslation ambiguity of the classical charge is defined and  supertranslation invariance of the invariant charges is proved. Section 4 demonstrates the boost invariance for both the classical charge and the invariant charge. Section 5 and 6 generalize the results in Section 3 and 4 to extended BMS fields. There are three appendices presenting useful formula and explicit forms of extended BMS fields.

\section{Background information}
In this section, we review the Bondi-Sachs coordinate systems and compute the evolution of the classical charge in such a coordinate system.
\subsection{Bondi-Sachs coordinates and BMS transformations}
In terms of a Bondi-Sachs coordinate system $(u, r,  x^2, x^3)$, near $\mathscr{I}^+$ of a vacuum spacetime, the metric takes the form
\begin{equation}\label{spacetime_metric}g_{\alpha\beta}dx^\alpha dx^\beta= -UV du^2-2U dudr+r^2 h_{AB}(dx^A+W^A du)(dx^B+W^B du).\end{equation} 

The index conventions here are $\alpha, \beta=0,1, 2, 3$, $A, B=2, 3$, and $u=x^0, r=x^1$. See \cite{BVM, MW} for more details of the construction of the coordinate system. 

The metric coefficients $U, V, h_{AB}, W^A$  of \eqref{spacetime_metric} depend on $u, r, \theta, \phi$, but $\det h_{AB}$ is independent of $u$ and $r$. These gauge conditions thus reduce the number of metric coefficients of a Bondi-Sachs coordinate system to six (there are only two independent components in $h_{AB}$). On the other hand, the boundary conditions $U\rightarrow 1$, $V\rightarrow 1$, $W^A\rightarrow 0$, $h_{AB}\rightarrow \sigma_{AB}$ are imposed as $r\rightarrow \infty$. The special gauge choice of the Bondi-Sachs coordinates implies a hierarchy among the vacuum Einstein equations, see \cite{ HPS,MW}.

 Assuming the outgoing radiation condition \cite{BVM, Sachs, MW}, the boundary condition and the vacuum Einstein equation imply that as $r\rightarrow \infty$, all metric coefficients can be expanded in inverse integral powers of $r$. In particular (see Chrusciel-Jezierski-Kijowski \cite[(5.98)-(5.100)]{CJK} for example), 
\[\begin{split} U&=1 - \frac{1}{16r^2} |C|^2 + O(r^{-3}),\\
V&=1-\frac{2m}{r}+ \frac{1}{r^2}\lt( \frac{1}{3}\na^A N_A + \frac{1}{4} \na^A C_{AB} \na_D C^{BD} + \frac{1}{16}|C|^2 \rt) + O(r^{-3}),\\
W^A&= \frac{1}{2r^2} \na_B C^{AB} + \frac{1}{r^3} \lt( \frac{2}{3}N^A - \frac{1}{16} \na^A |C|^2 -\frac{1}{2} C^{AB} \na^D C_{BD} \rt) + O(r^{-4}),\\
h_{AB}&={\sigma}_{AB}+\frac{C_{AB}}{r}+ \frac{1}{4r^2} |C|^2 \sigma_{AB} + O(r^{-3})\end{split} \] where  $m=m(u, x^A)$ is the mass aspect, $N_A = N_A(u,x^A)$ is the angular aspect and $C_{AB}=C_{AB}(u, x^A)$ is the shear tensor of this Bondi-Sachs coordinate system. Note that our convention of angular momentum aspect differs from that of Chrusciel-Jezierski-Kijowski, $N_A =-3 N_{A(CJK)}$. Here we take norm, raise and lower indices of tensors with respect to the metric $\sigma_{AB}$. We also define the {\it news tensor} $N_{AB} = \pl_u C_{AB}$.

To close this subsection, we note that, 
identifying a BMS field $Y=(Y^0+uY^1)\partial_u+Y^A \partial_A$ with $(Y^A,Y^0)$, the Lie bracket of the BMS algebra is (see \cite[(2.11)]{BT2})
\[
[(Y_1^A,Y_1^0),(Y_2^A,Y_2^0)]=(\hat Y^A,\hat Y^0)
\]
where
\[
\begin{split}
\hat Y^A=&Y_1^B\na_BY_2^A -Y_2^B\na_BY_1^A\\
\hat Y^0=&Y_1^A\na_AY^0_2-Y_2^A\na_AY^0_1+ \frac{1}{2} \Big (Y^0_1\na_A Y^A_2  -Y^0_2\na_A Y^A_1 \Big )
\end{split}
\]
In particular, the Lie bracket between a supertranslation $f \partial_u$ and a BMS field $Y$ is 
\begin{equation}\label{supertranslation_bracket}
[ f \partial_u, Y]=(\frac{1}{2}f \nabla_AY^A -Y^A \nabla_A f  )\partial_u.
\end{equation}
As a result, the supertranslations form an ideal in the BMS algebra.

\subsection{Evolution of the classical charge}
In this subsection, we derived the evolution formula of the classical charge, which will be used to calculate the total flux of the classical charge and the invariant charge in the later sections. 

The  evolution of the mass aspect \cite[(5.102)]{CJK} and the angular momentum aspect \cite[Proposition 3.1]{CKWWY} are given by 
\begin{equation}\label{dum}\partial_u m=-\frac{1}{8}N_{AB} N^{AB}+ \frac{1}{4} \nabla^A\nabla^B N_{AB},
\end{equation}
and \begin{equation}\label{duNA}\begin{split}\partial_u N_A
=& \na_A (m+\frac{1}{8}C_{BE} N^{BE}) + \frac{1}{4} \epsilon_{AB} \nabla^B( \epsilon^{PQ} \nabla_P \nabla^E C_{EQ})\\
&+\frac{1}{8} \epsilon_{AB} \nabla^B (\epsilon^{PQ} C_P^{\,\,\, E}N_{EQ} )     +\frac{1}{2} C_{AB}\nabla_D N^{DB}.\end{split}\end{equation}

 \begin{proposition}\label{evolution1} Let $Y$ be an extended BMS field defined on $(-\infty, \infty)\times \mathscr{U}$. The classical  charge $ \tilde{Q}(u, Y)$ evolves according to the following:
\begin{equation}\label{charge_evolution}\begin{split}
\partial_u \tilde Q(u, Y) &=\frac{1}{4}\int_{\mathscr{U}} Y^A \lt[ C_{AB}\nabla_{D}N^{BD} -N_{AB}\nabla_{D}C^{BD} +\frac{1}{2} \epsilon_{AB} \nabla^B (\epsilon^{PQ} C_P^{\,\,\, E}N_{EQ})        \rt]\\
&\quad -\frac{1}{4}\int_{\mathscr{U}} uY^1|N|^2 +\int_{\mathscr{U}}  Y^0 \partial_u (2m)\\
&\quad +\frac{1}{4}\int_{\mathscr{U}}  u \lt[  (2Y^1)\na^A\na^BN_{AB} -   (\epsilon_{AB} \nabla^A Y^B) \epsilon^{PQ} \nabla_P \nabla^E N_{EQ} \rt]
\end{split}\end{equation}
Note the last line vanishes for classical BMS fields.
\end{proposition}
\begin{comment}
\begin{remark}
For classical fields, the evolution of the classical charge is 
\begin{equation}\label{charge_evolution_classical}\begin{split}
\partial_u \tilde Q(u, Y) &= \frac{1}{4}\int_{S^2} Y^A \lt[ C_{AB}\nabla_{D}N^{BD} -N_{AB}\nabla_{D}C^{BD} +\frac{1}{2} \epsilon_{AB} \nabla^B (\epsilon^{PQ} C_P^{\,\,\, E}N_{EQ})        \rt]\\
&\quad -\frac{1}{4}\int_{S^2} uY^1|N|^2 +\int_{S^2}  Y^0 \partial_u (2m)
\end{split}\end{equation}
\end{remark}
\end{comment}
\begin{proof}
From the definition of the classical charge \eqref{classical}, we compute 
\begin{equation}\begin{split}\label{classical_evolution}
\partial_u\tilde Q(u, Y) =& \int_{\mathscr{U}} Y^A [\partial_u N_A-\frac{1}{4}\partial_u(C_{A}^{\,\,\,\,D}\nabla^B C_{DB}) - \frac{1}{16}\nabla_{A}\partial_u (C_{DE}C^{DE})]\\
&+ \int_{\mathscr{U}} Y^0 \partial_u(2m)+\int_{\mathscr{U}} Y^1 (2m)+\int_{\mathscr{U}} uY^1 \partial_u(2m).\\
&- \frac{1}{4} \int_\mathscr{U}  (\epsilon_{AB} \nabla^A Y^B) \epsilon^{PQ} \nabla_P \nabla^E C_{EQ} + u (\epsilon_{AB} \nabla^A Y^B) \epsilon^{PQ} \nabla_P \nabla^E N_{EQ} 
\end{split}\end{equation} and then apply \eqref{dum} and \eqref{duNA}. Finally, we use \eqref{cke}$,\nabla_A Y^A=2Y^1$,  to show that $\int_{S^2} \lt[Y^A\na_A m+Y^1(2m) \rt]=0$. \end{proof}

%\begin{proof}
%From the definition of the classical charge \eqref{classical} and \eqref{duNA}, we compute 
%\begin{equation}\begin{split}\label{classical_evolution}
%\partial_u\tilde Q(u, Y) &= \frac{1}{4}\int_{S^2} Y^A \lt[ C_{AB}\nabla_{D}N^{BD} -N_{AB}\nabla_{D}C^{BD} +\frac{1}{2} \epsilon_{AB} \nabla^B (\epsilon^{PQ} C_P^{\,\,\, E}N_{EQ}) \rt]\\
%&\quad -\frac{1}{4} \int_{S^2} \epsilon_{AB} \na^B Y^A \epsilon^{PQ} \na_P\na^E C_{EQ}\\
%&\quad + \int_{S^2} Y^0 \partial_u(2m)+\int_{S^2} Y^A\na_A m +  Y^1 (2m) +\int_{S^2} uY^1 \partial_u(2m).
%\end{split}\end{equation} Recall $\nabla_A Y^A=2Y^1$ and hence $\int_{S^2} Y^A\na_A m+Y^1(2m) =0$.

%For classical BMS fields, $\epsilon_{AB} \na^B Y^A$ and $Y^1$ are either zero or a (-2) eigenfunction. Integration by parts and \eqref{dum} imply the second line vanishes and $\int_{S^2} uY^1 \pl_u(2m) = -\frac{1}{4} \int_{S^2} %uY^1 |N|^2$. This completes the proof. \end{proof}

\subsection{Non-radiative spacetimes}
By a {\it non-radiative spacetime}, we mean a spacetime with zero news tensor $N_{AB}(u,x) \equiv 0$. This includes all model spacetimes such as Minkowski and Kerr.
The vanishing of the news tensor implies $\pl_u m(u, x)=0, \pl_u C_{AB}(u, x)=0$ and thus 

\[ m(u, x)\equiv \mathring{m}(x), C_{AB}(u, x)\equiv \mathring{C}_{AB}(x)\] and both potentials $c$ and $\underline{c}$ are independent of $u$ as well. 
By Proposition \ref{evolution1},  the classical charge is independent of $u$ in a non-radiative spacetime. In fact, the same holds for the invariant charge since the corrections terms are
independent of $u$ in a non-radiative spacetime as well.
\begin{comment}
\begin{lemma}
Suppose the news $N_{AB}(u, x)\equiv 0$ in a Bondi-Sachs coordinate system $(u, x)$, then the invariant charge $Q(u, Y) $ is constant, i.e. independent of the retarded time $u$. 
\end{lemma}

\begin{proof} 
Since $c$ and $m$ are both independent of $u$, We observe using Definition \ref{def_invariant_charge} that 
\[
\partial_u Q(u, Y) =\int_{\mathscr{U}}Y^A \partial_u N_A + 2 Y^1m -\frac{1}{4}(\epsilon_{AB} \nabla^A Y^B) \epsilon^{PQ} \nabla_P \nabla^E C_{EQ}
\]
When the news tensor vanishes, we have 
\[\partial_u N_A  (u, x)=  \na_A m + \frac{1}{4} \epsilon_{AB} \nabla^B( \epsilon^{PQ} \nabla_P \nabla^E C_{EQ})\]
The lemma follows immediately from integration by part and using $2Y^1=\nabla_AY^A$.
\end{proof}
\end{comment}

In the rest of this subsection, we pin down the exact formulae for the transformation of Bondi-Sachs data under supertranslations and boost, in a non-radiative spacetime. In general, under a BMS transformation the shear and news tensors are transformed by \cite[Corollary 2.7]{CKWWY2} 
\begin{align}
\bar C_{ab}(\bar u,\bar x) &= K^{-1}(\bar x) \pl_a g^A \pl_b g^B C_{AB}(\hat f, g(\bar x)) - 2 \bna_a \bna_b \hat f + \bar\Delta \hat f \bar\sigma_{ab} \label{shear}\\
\bar N_{ab}(\bar u, \bar x) &= \pl_{\bar u} \bar C_{ab} = \pl_a g^A \pl_b g^B N_{AB}(\hat f, g(\bar x)) \label{news}
\end{align}
where $\hat f(\bar u, \bar x) = K(\bar x)\bar u + f(\bar x)$. See also \cite[page 163]{CJK}. Consequently, vanishing of news is preserved under BMS transformations.

First we focus on supertransltion.  Since the spherical coordinate is unchanged, we use the same index $A$ for both $x$ and $\bar x$ to simplify the notation and   use $x$ to denote either $x^A$ or $\bar x^A$ in the following lemma.
\begin{lemma}
Suppose the news $N_{AB}(u, x)\equiv 0$ in a Bondi-Sachs coordinate system $(u, x)$ and $(\bar{u}, x)$ is another Bondi-Sachs coordinate system that is related to $(u, x)$ by a supertranslation $u=\bar{u}+f$ for $f\in C^\infty(S^2)$. Then we have
\begin{equation}\label{transform}\begin{split}
\bar{m}(\bar u,x) &= \mathring{\bar{m}}(x)= \mathring{m}(x) \\
\bar{C}_{AB}(\bar u,x) &= \mathring{\bar{C}}_{AB}(x)=\mathring{C}_{AB}(x) - F_{AB} \\
\mathring{\bar{c}}&=\mathring{c}-2 f_{\ell\geq 2}\\
\mathring{\bar{\underline{c}}}&=\mathring{\underline{c}}
\end{split}\end{equation} where $F_{AB}=2 \na_A\na_B f -\Delta f \sigma_{AB} $
and
\begin{equation}
\label{AM_aspect2} 
\begin{split}
\bar{N}_A  (\bar{u}, x)=  N_A  (u_0, x) +(\bar{u}-u_0+f)       (\na_A \mathring{m}-\frac{1}{4}\nabla^B \mathring{P}_{BA}  )+3\mathring{m}\nabla_A f-\frac{3}{4}\mathring{P}_{BA} \na^B f
\end{split}
\end{equation} for any $\bar{u}$ and fixed $u_0$. 
\end{lemma}
\begin{proof}
The mass aspect $\bar{m} (\bar u, x)
$, the shear $\bar{C}_{AB}(\bar u,x)$, and the news $ \bar{N}_{AB}(\bar u,x)$  in the $(\bar u, \bar{x})$ coordinate system are related to the mass aspect ${m} (u, x)
$, the shear ${C}_{AB}(u,x)$, and the news $ {N}_{AB}(u,x)$  in the $(u, x)$ coordinate system through:
\begin{equation}\label{basic_data}\begin{split}
\bar{m}(\bar u,x) &= m(\bar u+f,x) + \frac{1}{2} (\na^B N_{BD})(\bar u+f,x) \na^D f \\
& + \frac{1}{4} (\partial_{u}N_{BD})(\bar u+f,x) \na^B f\na^D f + \frac{1}{4} N_{BD}(\bar u+f,x) \na^B\na^D f\\
\bar{C}_{AB}(\bar u,x) &= C_{AB}(\bar u+ f(x),x) - 2 \na_A\na_B f + \Delta f \sigma_{AB} \\
\bar{N}_{AB}(\bar u,x) &= N_{AB}(\bar u+f(x),x)
\end{split}\end{equation}
See \cite[(C.117) and (C.119)]{CJK} for example. 

In particular, $N_{AB}(u, x)\equiv 0$ implies that $\bar{N}_{AB}(\bar u,x)\equiv 0$ and
\begin{align*}
\bar{m}(\bar u,x) &= \mathring{\bar{m}}(x)= \mathring{m}(x) \\
\bar{C}_{AB}(\bar u,x) &= \mathring{\bar{C}}_{AB}(x)=\mathring{C}_{AB}(x) - F_{AB} \\
\mathring{\bar{c}}&=\mathring{c}-2 f_{\ell\geq 2}\\
\mathring{\bar{\underline{c}}}&=\mathring{\underline{c}}.
\end{align*}

For the angular momentum aspect, we have
\[\partial_u N_A  (u, x)= \na_A m (u, x) -\frac{1}{4}\nabla^B P_{BA} (u, x)\]
where \[ P_{BA} (u, x)=(\nabla_B \nabla^E C_{EA}-\nabla_A \nabla^E C_{EB})(u, x).\]

Therefore, \[\partial_u N_A  (u, x)= \na_A \mathring{m} (x) -\frac{1}{4}\nabla^B \mathring{P}_{BA} (x)\] is independent of $u$. Integrating gives
\begin{equation}\label{AM_aspect1} N_A  (u, x)=  N_A  (u_0, x) +(u-u_0)       (\na_A \mathring{m} -\frac{1}{4}\nabla^B \mathring{P}_{BA} )\end{equation} for any $u$ and fixed $u_0$.

Finally, from  \cite{CKWWY2} the angular momentum aspect transforms by 
\[\begin{split}\bar{N}_A(\bar{u}, x)&=N_{A}(\bar{u}+f, x)+3m(\bar{u}+f, x)\nabla_A f-\frac{3}{4}(\nabla_B \nabla^E C_{EA}-\nabla_A \nabla^E C_{EB})(u, x)
\na^B f\\
&=N_{A}(\bar{u}+f, x)+3\mathring{m}\nabla_A f-\frac{3}{4}\mathring{P}_{BA} \na^B f.\\
\end{split}\]

Combining with \eqref{AM_aspect1} and setting $u=\bar{u}+f$, we obtain \eqref{AM_aspect2}.
\end{proof}
Next we focus on the boost. 
\begin{lemma}\label{transform_boost}
Suppose the news $N_{AB}(u,x) \equiv 0$ in a Bondi-Sachs coordinate system $(u,x)$, and $(\bar u, \bar x)$ is another Bondi-Sachs coordinate system that is related to $(u,x)$ by a boost $B_g$, namely,
\begin{align}
u = K(\bar x) \bar u, x = g(\bar x).
\end{align}
Then we have
\begin{align*}
\begin{split}
\bar m(\bar u, \bar x) &= K^3 \mathring m \\
\bar C_{ab} (\bar u, \bar x) &= K^{-1} \pl_a g^A \pl_b g^B \mathring C_{AB}\\
\bar c(\bar u,\bar x) &= K^{-1} \mathring c\\
\bar N_a (\bar u, \bar x) &= K^2 \pl_a g^A \lt(  N_A  (u_0, g(\bar x)) +(K\bar u-u_0) (\na_A \mathring{m} -\frac{1}{4}\nabla^B \mathring{P}_{BA} )\rt) \\
&\quad + 3 \mathring m K^2 \pl_a K \bar u - \frac{3}{4} K \pl_a g^A \mathring P_{BA} K \hat\na^b K \bar u \pl_b g^B
\end{split}
\end{align*}
where $K$ is evaluated at $\bar x$ and $\mathring m, \mathring C_{AB},\mathring c, \mathring P_{BA}$ are evaluated at $g(\bar x)$.
\end{lemma}
\begin{proof}
 The formula for $\bar C_{ab}$ follows from $C_{AB}(K\bar u, g(\bar x)) = \mathring C_{AB}(g(\bar x))$ when $N_{AB} \equiv 0$ and then the formula for $\bar c$ follows from \eqref{bar_traceless_Hessian}.
 
Finally, from \cite[Theorem 3.1 and Theorem 3.3]{CKWWY2}, the assumption $N_{AB} \equiv 0$ implies that 
\begin{align*}
\bar m(\bar u, \bar x) &= K^3 m\\
\bar N_a (\bar u, \bar x) &= K^2 \pl_a g^A N_A + 3 m K^2 \pl_a K \bar u - \frac{3}{4} K \pl_a g^A P_{BA} K \hat\na^b K \bar u \pl_b g^B
\end{align*}
where $K$ is evaluated at $\bar x$ and $N_A, m, P_{BA}$ are evaluated at $(K\bar u, g(\bar x))$. The formula follows from \eqref{AM_aspect1}.
\end{proof} 
\part{Classical BMS fields}
\section{Supertranslation invariance of the invariant charges}\label{sec:supertranslation_invariance_classical}
In this section, we revisit the study of the effect of supertranslation on the total flux of a classical BMS field in \cite{CKWWY}, providing a unified treatment of angular momentum and center of mass integral. In the first subsection, the supertranslation ambiguity of the classical charge is identified, and then supertranslation invariance of the invariant charge is shown in the next subsection. In the third subsection we show that the classical charge itself, not just its total flux, is supertranslation invariant in non-radiative spacetimes.

\subsection{Supertranslation ambiguity of the classical charge}\label{subsec:supertranslation_ambiguity}
We study the effect of supertranslation on the total flux of the classical charge along null infinity or, equivalently, the difference of classical charge at timelike infinity and spatial infinity. Suppose $I=(-\infty,\infty)$ and $\mathscr{I}^+$ is complete extending from spatial infinity ($u=-\infty$) to timelike infinity ($u=+\infty$). A supertranslation is a change of coordinates $(\bar{u}, \bar x^A)\rightarrow (u, x^a)$ such that $u = \bar u + f(x), x^A = \delta^A_a\bar x^a$ on $\mathscr{I}^+$. Let $m$, $C_{AB}$, and $N_{AB}$ denote the mass aspect, the shear, and the news, respectively, in the $(u,  x^A)$ coordinate system. Since the spherical coordinate is unchanged, we use the same index $A$ for both $x$ and $\bar x$ to simplify the notation and   use $x$ to denote either $x^A$ or $\bar x^A$ throughout this section. 

We assume that there exists a constant $\varepsilon>0$ such that
\begin{align}\label{news_decay_1}
N_{AB}( u,x) = O(|u|^{-1-\varepsilon}) \mbox{ as } u \rw \pm\infty.
\end{align}
The decay rate \eqref{news_decay_1} implies that the limits of the shear tensor and the mass aspect exist 
\[ \lim_{ u \rw \pm\infty} C_{AB} (u,x) = C_{AB}(\pm), \lim_{ u \rw \pm\infty} m (u,x) = m(\pm).\]

Recall the mass aspect $\bar{m} (\bar u, x)$, the shear $\bar{C}_{AB}(\bar u,x)$, and the news $ \bar{N}_{AB}(\bar u,x)$  in the $(\bar u, x)$ coordinate system are related to those in  the $( u, x)$ coordinate system through:
\begin{equation}\label{supertranslation}\begin{split}
\bar{m}(\bar u,x) &= m(\bar u+f,x) + \frac{1}{2} (\na^B N_{BD})(\bar u+f,x) \na^D f \\
&\quad + \frac{1}{4} (\partial_{u}N_{BD})(\bar u+f,x) \na^B f\na^D f + \frac{1}{4} N_{BD}(\bar u+f,x) \na^B\na^D f\\
\bar{C}_{AB}(\bar u,x) &= C_{AB}(\bar u+ f(x),x) - 2 \na_A\na_B f + \Delta f \sigma_{AB} \\
\bar{N}_{AB}(\bar u,x) &= N_{AB}(\bar u+f(x),x)
\end{split}\end{equation}
Applying the chain rule on \eqref{supertranslation} yields
\begin{align*}
\na_D \bar C_{AB} (\bar u,x)&= N_{AB}(\bar u + f,x)\na_D f + (\na_D C_{AB})(\bar u+f,x) -\na_D F_{AB},\\ 
\na_D \bar{C}^{BD}(\bar u,x) &= N^{BD}(\bar u+f,x) \na_D f + (\na_D C^{BD}) ( \bar u+f ,x) -\na_D F^{BD}, 
\end{align*} where
 \[ F_{AB} = 2\na_A\na_B f - \Delta f \sigma_{AB} \] is a $u$-independent symmetric traceless two-tensor on $S^2$. 
 
We are ready to identify the supertranslation ambiguity of the total flux of classical charge for classical BMS fields.   
\begin{theorem}\label{supertranslation_ambiguity_classical} Suppose $Y$ is a classical BMS field and $N_{AB}=O(|u|^{-1-\epsilon})$, then the supertranslation ambiguity of the total flux of the classical charge $\delta \tilde{Q}$ is 
\[
   \delta \tilde Q(\bar Y) - \delta \tilde Q(Y)  = \int_{S^2}(-f \nabla_AY^A +2Y^A \nabla_A f )(m(+)- m(-)),
\]
if $\bar{Y}$ is related to $Y$ by a supertranslation $S_f$.
\end{theorem}

\begin{proof}

By Proposition \ref{evolution1} and the remark afterward, we obtain 
\begin{equation}\begin{split}
&\delta \tilde{Q}(Y)-\int 2Y^0(m(+)-m(-))\\
=& \frac{1}{4}\int_{-\infty}^{+\infty}\int_{S^2} Y^A \lt( C_{AB}\nabla_{D}N^{BD} -N_{AB}\nabla_{D}C^{BD} +\frac{1}{2} \epsilon_{AB} \nabla^B (\epsilon^{PQ} C_P^{\,\,\, E}N_{EQ}  )             \rt)\\
-&\frac{1}{4}\int_{-\infty}^{+\infty}\int_{S^2} uY^1|N|^2 =\RN{1}+\RN{2}.
\end{split}\end{equation}
Suppose $\bar{Y}$ is related to $Y$ by a supertranslation $S_f$, we derive:
\begin{equation}\label{after_classical}\begin{split}
&\delta \tilde{Q}(\bar{Y})-\int 2Y^0(m(+)-m(-))\\
=& \frac{1}{4}\int_{-\infty}^{+\infty}\int_{S^2} Y^A \lt( \bar C_{AB}\nabla_{D}\bar{N}^{BD} -\bar{N}_{AB}\nabla_{D}\bar{C}^{BD} +\frac{1}{2} \epsilon_{AB} \nabla^B (\epsilon^{PQ} \bar{C}_P^{\,\,\, E}\bar{N}_{EQ} )\rt)\\
-&\quad \frac{1}{4}\int_{-\infty}^{+\infty}\int_{S^2} \bar{u}Y^1|\bar N|^2 =\bar{\RN{1}}+\bar{\RN{2}}.\\
\end{split}\end{equation}
To evaluate $\bar{\RN{1}}-\RN{1}$, we compute 
\begin{align*}
\nabla^B ( \bar{C}_P^{\,\,\, E}\bar{N}_{EQ} )& =     \partial_u (({C}_P^{\,\,\, E}-F_P^{\,\,\, E}){N}_{EQ}  \na^B f)+ \na_B( ({C}_P^{\,\,\, E}-F_P^{\,\,\, E}){N}_{EQ} ) \\
\bar{C}_{AB} \na_D \bar{N}^{BD}& = \partial_u ((C_{AB}-F_{AB}) N^{BD} \na_D f)-N_{AB}N^{BD}\nabla_Df + (C_{AB}-F_{AB})\na_D N^{BD} \\
\bar{N}_{AB} \na_D \bar{C}^{BD} &= N_{AB} N^{BD} \na_D f + N_{AB}(\na_D C^{BD}  -\na_D F^{BD}).\\
\end{align*}
All left-hand sides are evaluated at $(\bar{u}, x)$ and all right hand sides are evaluated at $(\bar{u}+f, x)$. 
Plugging these into \eqref{after_classical}, applying Fubini's theorem and change of variable, the assumption \eqref{news_decay_1}, the identity $2N_{AB}N^{BD}=|N|^2\delta_A^D$, and \eqref{integral_formula}, we obtain
\[\begin{split}&\bar{\RN{1}}-\RN{1}\\
&=\frac{1}{4}\int_{-\infty}^{+\infty}\int_{S^2}\lt[ Y^A \lt( -F_{AB}\nabla_{D}N^{BD} +N_{AB}\nabla_{D}F^{BD}-\frac{1}{2} \epsilon_{AB} \nabla^B (\epsilon^{PQ} F_P^{\,\,\, E}N_{EQ} )  -\na_A f |N|^2\rt) \rt]  \\
&=\frac{1}{4}\int_{-\infty}^{+\infty} \int_{S^2}(-\nabla_AY^A f + 2 Y^A \nabla_A f )(\nabla_B\nabla_D N^{BD})- Y^A\na_A f |N|^2.\\
\end{split}\]

On the other hand,
\[\begin{split}\bar{\RN{2}}-\RN{2}&=-\frac{1}{4}\int_{-\infty}^{+\infty} \int_{S^2} (-f)(\frac{1}{2}\na_A Y^A)|{N}|^2)= \frac{1}{4}\int_{-\infty}^{+\infty} \int_{S^2} \frac{1}{2}f(\na_A Y^A)|{N}|^2. \end{split}\]

Therefore $\bar{\RN{1}}-\RN{1}+\bar{\RN{2}}-\RN{2}$ is given by

\[ \bar{\RN{1}}-\RN{1}+\bar{\RN{2}}-\RN{2}=\int_{-\infty}^{\infty}\int_{S^2}(-\nabla_AY^A f + 2 Y^A \nabla_A f )(\partial_u m). \]
The theorem follows immediately.
\end{proof}

 \begin{remark}
The case $Y^A = \epsilon^{AB}\na_B \tilde X^k$ of Theorem \ref{supertranslation_ambiguity_classical} corresponds to the combination of (3.21) and (3.22) in \cite{AKL}.
\end{remark}

\begin{corollary}\label{supertranslation invariance of total flux} Suppose the news tensor decays as
\[ N_{AB}( u,x) = O(|u|^{-1-\varepsilon}) \mbox{ as } u \rw \pm\infty,
\]
then the total flux of the classical charge is supertranslation invariant for all classical BMS fields if and only if $m(+)-m(-)$ is a constant on $S^2$. 

\end{corollary}

\begin{proof} The supertranslation invariance is equivalent to 
\[\int_{S^2}f \lt[-3\nabla_AY^A (m(+)-m(-)) -2 Y^A \nabla_A  (m(+)-m(-))\rt]=0, \] for any function $f$ of mode $\ell\geq 2$ and any conformal Killing field $Y^A$ on $S^2$.  This is equivalent to the condition that 
\[ -3\nabla_AY^A (m(+)-m(-)) -2 Y^A \nabla_A  (m(+)-m(-))\] must be a function of mode $\ell\leq 1$ for any conformal Killing field $Y^A$, or $m(+)-m(-)$ must be of $\ell=0$. 
\end{proof}

\begin{remark}
Let
\[ (u,x) = B(\bar u, \bar x) = (K(\bar x) \bar u + f(\bar x), g(\bar x)) \]
be a BMS transformation and choose local coordinates $\bar x^a$ and $x^A$. Given a BMS field $\bar Y = (\bar Y^0 + \bar u \bar Y^1) \frac{\pl}{\pl \bar u} + \bar Y^a \frac{\pl}{\pl \bar x^a}$, we compute the differential of $B$
\begin{align*}
B_*(\bar Y) &= \begin{bmatrix}
K & \pl_a K \bar u + \pl_a f \\
0 & \pl_a g^A
\end{bmatrix} \begin{bmatrix}
\bar Y^0 + \bar u \bar Y^1 \\ \bar Y^a
\end{bmatrix} \\
&= \lt( K\bar Y^0 - f\bar Y^1 - K^{-1} f \bar Y^a \pl_a K + u(\bar Y^1 +K^{-1}\bar Y^a\pl_a K) + \bar Y^a \pl_a f \rt)\frac{\pl}{\pl u}\\
&\quad  + \pl_a g^A \bar Y^a \frac{\pl}{\pl x^A}.
\end{align*}
In view of this and \eqref{supertranslation_bracket}, Theorem \ref{supertranslation_ambiguity_classical} says $\delta\tilde Q(\bar Y) = \delta\tilde Q( {S_f}_* (\bar Y))$.
\end{remark}

\subsection{Supertranslation invariance of the total flux of the invariant charge}
We show that the invariant charge is supertranslation invariant for any classical BMS field. 
\begin{theorem}\label{classical_invariant}
Suppose $Y$ is a classical BMS field and $N_{AB}=O(|u|^{-1-\epsilon})$, then the total flux of the invariant charge is supertranslation invariant. Namely,
\[
\begin{split}
 \delta Q(\bar Y)- \delta Q(Y)=&\int_{S^2}(-f_{l\le 1} \nabla_AY^A +2Y^A \nabla_A f_{l\le 1}  )(m(+)- m(-))\\
                                              =&-\delta Q ([ f_{\ell\leq 1} \partial_u, Y])
\end{split}
\]
if $\bar{Y}$ is related to $Y$ by a supertranslation $S_f$ \eqref{supertranslation_related}.
\end{theorem}
\begin{proof}
Recall that for a classical BMS field $Y$ \eqref{invariant_charge_classical},
\[
Q(Y)= \tilde Q(Y)+ \int_{S^2} (Y^A \na_A c - Y^1 c )m.
\]
Suppose $Y$ and $\bar Y$ are related by a supertranslation $S_f$, then we have
\[
\begin{split}
    & \delta Q(\bar Y)- \delta Q(Y)\\
= & \delta \tilde Q(\bar Y)- \delta \tilde Q(Y)+ \lt[ \int_{S^2} (Y^A \na_A \bar c -Y^1 \bar c ) \bar m \rt]_{-\infty}^{\infty} - \lt[ \int_{S^2} (Y^A \na_A  c -Y^1 c ) m \rt]_{-\infty}^{\infty}\\
\end{split}
\]
After the supertranslation, we have $\bar c(\pm) = c(\pm) - 2f_{l\ge 2}$, $\underline {\bar{ c}}(\pm) = \underline c(\pm) $  and $\bar m(\pm) = m(\pm)$ from \eqref{supertranslation}. Recall $Y^1 = \frac{1}{2}\na_A Y^A$ and it follows that 
\[
 \delta Q(\bar Y)- \delta Q(Y)=\int_{S^2}(-f_{l\le 1} \nabla_AY^A +2Y^A \nabla_A f_{l\le 1}  )(m(+)- m(-))
\]
The result follows from \eqref{supertranslation_bracket}.
\end{proof}

\subsection{Non-radiative spacetimes}

Recall that both the classical and invariant charge are independent of the retarded time in a non-radiative spacetime.

\begin{comment}
\begin{lemma}
In non-radiative spacetimes, the invariant charge $Q(u, Y)$ of classical BMS fields is constant, i.e. independent of the retarded time $u$. 
\end{lemma}

\begin{proof} The assumption implies $\pl_u m(u, x)=0, \pl_u C_{AB}(u, x)=0$ and thus 

\[ m(u, x)\equiv \mathring{m}(x), C_{AB}(u, x)\equiv \mathring{C}_{AB}(x)\] and both potentials $c$ and $\underline{c}$ are independent of $u$ as well. 
Since $c$ and $m$ are both independent of $u$, We observe that 
\[
\partial_u Q(u, Y) =\int_{S^2}Y^A \partial_u N_A + 2 Y^1m -\frac{1}{4}(\epsilon_{AB} \nabla^A Y^B) \epsilon^{PQ} \nabla_P \nabla^E C_{EQ}
\]
Recall that when the news tensor vanishes, we have 
\[\partial_u N_A  (u, x)=  \na_A m + \frac{1}{4} \epsilon_{AB} \nabla^B( \epsilon^{PQ} \nabla_P \nabla^E C_{EQ})\]
The lemma follows immediately from integration by part and using $2Y^1=\nabla_AY^A$.
\end{proof}
\end{comment}

Fixing $\bar{u}=\bar{u}_0$, we consider the invariant charges
\[\begin{split}&Q=Q(u_0,Y)\\
&\bar Q=Q(\bar{u}_0,\bar Y)
\end{split}\]

We show that in non-radiative spacetimes, the invariant charge itself, not just its total flux, is invariant under supertranslation.
\begin{theorem} \label{invariant_CWY_vanish_news}
The invariant charge satisfies
\begin{align}
\bar{Q}-Q= \int_{S^2} (2 Y^A \na_A f_{\ell\le 1} - f_{\ell\le 1} \na_A Y^A)  \mathring m  
= -Q ([ f_{\ell\leq 1} \partial_u, Y])
\end{align}
\end{theorem}
\begin{proof}

Taking the difference of $\bar Q $ and  $Q$ and applying \eqref{transform}, we obtain 
 \begin{equation}
 \begin{split}
    &\bar Q-Q\\
 =&\frac{1}{4}\int_{S^2}  Y^A \Big [  \mathring{C}_{AB}\nabla_D F^{BD}+F_{AB} \nabla_D \mathring{C}^{BD}+\frac{1}{2} \nabla_A (\mathring{C}_{BD} F^{BD})- F_{AB} \nabla_D F^{BD} \\
      &\qquad\qquad\qquad -\frac{1}{4} \na_A( F_{BD}F^{BD} )   \Big ] +\int_{S^2} Y^A \lt[\bar{N}_A  (\bar{u}_0, x)-  N_A  (u_0, x)\rt]\\
& + \big[ \int_{S^2}(2Y^1 f_{\ell\ge 2} - 2Y^A \na_A f_{\ell\ge 2})  \mathring m\big] \\
&+ (\bar u_0 - u_0) \Big[ \int_{S^2} 2 \mathring m Y^1 -\frac{1}{4} \int_{S^2}  (\epsilon_{AB} \nabla^A Y^B) \epsilon^{PQ} \nabla_P \nabla^E \mathring{C}_{EQ} \Big]
 \end{split}
 \end{equation}
 We observe that,   by Lemma \ref{quadratic_integral},
  \[ \int_{S^2} Y^A \lt[  -F_{AB}\nabla_D F^{DB}-\frac{1}{4} \na_A( F_{BD}F^{BD} ) \rt]  = 0 .\] 
  
  We compute
\[\begin{split}&\int_{S^2} Y^A \lt[\bar{N}_A  (\bar{u}_0, x)-  N_A  (u_0, x)\rt] + (\bar u_0 - u_0) \Big[ \int_{S^2} 2  \mathring m Y^1 -\frac{1}{4} \int_{S^2}  (\epsilon_{AB} \nabla^A Y^B) \epsilon^{PQ} \nabla_P \nabla^E C_{EQ} \Big]
\\
=& \int_{S^2} Y^A\lt[     f(\na_A \mathring{m}+3\mathring{m}\nabla_A f-f \nabla^B \mathring{P}_{BA} -3\mathring{P}_{BA} \na^B f  
 \rt]\\ 
\end{split}
\]
As a result,
\[\begin{split}
   &\bar{Q}-Q\\
=&\frac{1}{4}\int_{S^2} Y^A \lt[  \mathring{C}_{AB}\nabla_D F^{BD}+F_{AB} \nabla_D \mathring{C}^{BD}+\frac{1}{2} \nabla_A (\mathring{C}_{BD} F^{BD})-f \nabla^B \mathring{P}_{BA} -3\mathring{P}_{BA} \na^B f   \rt]\\
&+ \int_{S^2}  (2 Y^A \na_A f_{\ell\le 1} - f_{\ell\le 1} \na_A Y^A) \mathring m
  \end{split}\]
By Lemma \ref{integral formula for nonradiative case}, the first line vanishes. This completes the proof using \eqref{supertranslation_bracket}.
\end{proof}

\section{Boost invariance of the classical and invariant charges}\label{sec:boost_classical}
\subsection{Total flux of classical and invariant charges}
In this section, we show that the total flux of both  classical and invariant charge is invariant under a boost. Suppose $g: (S^2, \bar{\sigma}) \rw (S^2, \sigma)$ is a conformal map with conformal factor $K$. We denote $g^*\sigma = \hat\sigma = K^2 \bar\sigma$. Consider the boost $B_g$ in the BMS group. Denoting the coordinate by  $(\bar u, \bar x^a)$ and $(u,x^A)$, we have $B_g(\bar u, \bar x) = (u,x) $ where $u=K \bar u, x=g(\bar{x})$ and consider two BMS fields $\bar{Y}=(\bar{Y}^0+\bar{u}\bar{Y}^1)\partial_{\bar{u}}+\bar{Y}^a\partial_{a}$ and $Y= B_*(\bar Y) = (Y^0+uY^1)\partial_u+Y^A \partial_A$ related by the boost, Definition \ref{def_boost_related}. In particular, we have the relation of the spherical vectors $\pl_a g^A \bar Y^a = Y^A$ on $S^2$.

The main result of this subsection is the following.
\begin{theorem} \label{boost_invariance_flux}
Suppose $\bar Y$ and $Y$ are two classical BMS fields related by a boost and $N_{AB}=O(|u|^{-1-\epsilon})$ as $u\rightarrow \pm \infty$
. We have
\[
\delta \tilde Q(\bar Y) = \delta \tilde Q (Y)\quad\mbox{and}\quad  \delta Q(\bar Y) = \delta Q(Y).
\]
\end{theorem}
\begin{proof}
We first deal with the case of  the classical charge and separate it into two cases $\bar Y = \bar Y^0 \partial_{\bar u}$ and $\bar Y = \bar u \bar Y^1 \partial_{\bar u}+ \bar Y^a \partial_a$.\\
\noindent Case 1: Suppose  $\bar Y = \bar Y^0 \partial_{\bar u}$. From the transformation of the mass aspect function
\begin{align*}
\bar m(\pm) = K^3(\bar x) (m(\pm) \circ g)(\bar x),
\end{align*} we have
\begin{align*}
\delta\tilde Q(\bar Y) &= \lt. \int_{S^2} \bar Y^0(\bar x^a) (2\bar m) d\bar\sigma \rt]_{-\infty}^{+\infty}\\
&= \lt. \int_{S^2} K(\bar x) \bar Y^0(\bar x) (2m\circ g)(\bar x) d\hat\sigma \rt]_{-\infty}^{+\infty}\\
&= \lt. \int_{S^2} K(g^{-1}(x)) \bar Y^0(g^{-1}(x)) \cdot 2m(x) d\sigma \rt]_{-\infty}^{+\infty}\\
&= \delta\tilde Q(Y). 
\end{align*}
Case 2: Suppose $\bar Y = \bar u \bar Y^1 \partial_{\bar u}+ \bar Y^a \partial_a$. We have
\begin{align}\label{Ashtekhar-Streubel flux}
\begin{split}
4\delta\tilde Q(\bar Y) &= \int_{-\infty}^\infty \int_{S^2} \lt[ \bar Y^a \lt( \bar C_{ae} \bar\na^d N_{bd} \bar\sigma^{be} - \bar N_{ad} \bar\na^d \bar C_{bd} \bar\sigma^{be}\rt)  \rt] d\bar\sigma d\bar u\\
&\quad +\int_{-\infty}^\infty \int_{S^2} -\frac{1}{2} \bar\na_b \bar Y^a \bar\epsilon_a^{\;\;b} (\bar\epsilon^{pq} \bar C_{pe} \bar N_{dq} \bar\sigma^{de})\, d\bar\sigma d\bar u\\
&\quad +\int_{-\infty}^\infty \int_{S^2} -\bar u \cdot \frac{1}{2} \bar\na_a \bar Y^a \bar\sigma^{bc}\bar\sigma^{de} \bar N_{bd}\bar N_{ce}\, d\bar\sigma d\bar u.
\end{split}
\end{align}

By \eqref{shear} and \eqref{news}, the divergence of the shear tensor and the news tensor are given by
\begin{align}
\bar\na^d (K\bar C_{bd}) &= K^2 \hat\na^d (K\bar C_{bd}) = K^2 \lt( \pl_b g^B \na^D C_{BD} + \pl_b g^B \pl_d g^D N_{BD} \bar u \hat\na^d K \rt) \label{div shear}\\
\bar\na^d \bar N_{bd} &= K^2 \hat\na^d \bar N_{bd} = K^2 \lt( \pl_b g^B \na^D N_{BD} + \pl_b g^B \pl_d g^D \pl_u N_{BD} \bar u \hat\na^d K \rt). \label{div news}
\end{align}
Moreover, we have
\begin{align*}
\bar\na_b \bar Y^a = \hat\na_b \bar Y^a +K^{-1} \lt( -\pl_b K \bar Y^a - \pl_d K \bar Y^d \delta^a_b + \hat\na^a K \hat\sigma_{bd} \bar Y^d \rt)
\end{align*}
and
\begin{align*}
\bar\na_a \bar Y^a = \hat\na_a \bar Y^a -2K^{-1} \pl_a K \bar Y^a.
\end{align*}
Putting these together with the change of variable formula
\begin{align}\label{change of variable}
\int_{-\infty}^{+\infty}\int_{S^2} F(\bar u, \bar x) d\hat\sigma d\bar u = \int_{-\infty}^{+\infty}\int_{S^2} F(u,g^{-1}(x)) \cdot (K^{-1} \circ g^{-1})(x) d\sigma du,
\end{align}
we obtain
\begin{align*}
&4\delta\tilde Q(\bar Y) - 4\delta\tilde Q(Y) \\
&\quad= \int_{-\infty}^{+\infty}\int_{S^2} \pl_a g^A \bar Y^a C_A^B \pl_u N_{BD} u (K^{-1}\circ g^{-1}) \na^D (K \circ g^{-1}) \, d\sigma du \\
&\quad + \int_{-\infty}^{+\infty}\int_{S^2}- \pl_a g^A \bar Y^a N_A^B \lt( N_{BD} u \na^D (K \circ g^{-1}) - C_{BD} (K^{-1}\circ g^{-1}) \na^D (K\circ g^{-1})\rt) d\sigma du\\
&\quad +\int_{-\infty}^{+\infty}\int_{S^2} \pl_a g^A \bar Y^a (K^{-1}\circ g^{-1})\na^B (K\circ g^{-1})\epsilon_{AB} (\epsilon^{PQ} C_{P}^E N_{EQ}) d\sigma du\\
&\quad +\int_{-\infty}^{+\infty}\int_{S^2} \pl_a g^A \bar Y^a (K^{-1}\circ g^{-1}) \na_A (K\circ g^{-1}) u |N|^2 d\sigma du.
\end{align*}
The first integral on the right hand side vanishes in view of the decay condition of the news tensor and we obtain
\begin{align*}
  &4\delta\tilde Q(\bar Y) - 4\delta\tilde Q(Y) \\
&=  \int_{-\infty}^{+\infty}\int_{S^2} Y^A (K^{-1}\circ g^{-1}) \na^B (K\circ g^{-1}) \lt[ \epsilon_{AB} (\epsilon^{PQ} C_P^E N_{EQ}) - (C_A^D N_{DB} - C_B^D N_{DA}) \rt]\\
&= 0.  
\end{align*} 
This completes the proof of Case 2 and, by additivity, the assertion $\delta\tilde Q(\bar Y) = \delta\tilde Q(Y)$ for any two classical BMS fields $\bar Y$ and $Y$ related by a boost.

For the invariant charge, recall
\begin{align*}
\delta Q(Y) = \delta\tilde Q(Y) + \int_{S^2} (Y^A \na_A c - \frac{1}{2} \na_A Y^A c) m \Big]_{-\infty}^{+\infty}.
\end{align*}
The closed potential and mass aspect function at $u = \pm\infty$ are transformed by 
\begin{align*}
\bar c(\pm) (\bar x) &= K^{-1}(\bar x) c (g(\bar x)),\\
\bar m(\pm)(\bar x) &= K^3 (\bar x) m(g(\bar x)).
\end{align*}
At $u = \pm\infty$, we have
\begin{align*}
&\int_{S^2} ( \bar Y^a \na_a \bar c - \frac{1}{2} \bar \na_a \bar Y^a \bar c ) \bar m \, d\bar\sigma \\
&= \int_{S^2} \lt[ \bar Y^a \hat\na_a (K^{-1} c \circ g) - \frac{1}{2}\lt( \hat\na_a \bar Y^a - 2K^{-1} \hat\na_a K \bar  Y^a \rt) K^{-1} c \circ g \rt] K^3 m \circ g K^{-2} d\hat\sigma\\
&= \int_{S^2} ( \bar Y^a \hat\na_a (c\circ g) - \frac{1}{2}\hat\na_a \bar Y^a c\circ g ) m\circ g d\hat\sigma\\
&= \int_{S^2} \lt[ Y^A \na_A c - \frac{1}{2} \na_A Y^A c \rt] m d\sigma. 
\end{align*}
Hence the correction term is invariant under boost. This completes the proof of $\delta Q(\bar Y) = \delta Q(Y)$.
\end{proof}

 \begin{remark}
The above argument actually shows that the expression \eqref{Ashtekhar-Streubel flux} is invariant under boost for any vector field  $\bar{Y}^a$ on $S^2$.
\end{remark}
\begin{remark}
Using the notation $(Y^0,Y^A)$ for BMS fields, we define the components of angular momentum
\begin{align}
J^i(u) = Q(u, (0, \epsilon^{AB}\na_B \tilde X^i))
\end{align}
and center of mass integral
\begin{align*}
C^k(u) = Q(u, (0, \na^A \tilde X^k)).
\end{align*}
Suppose the conformal map $g$ is given by $z \mapsto \exp(\beta) z$ in complex coordinates, then Theorem \ref{boost_invariance_flux} reads
\begin{align*}
\delta \bar J^1 &= (\cosh\beta) \delta J^1 + (\sinh\beta) \delta C^2\\
\delta \bar J^2 &= (\cosh\beta) \delta J^2 - (\sinh\beta) \delta C^1\\
\delta \bar J^3 &= \delta J^3 
\end{align*}
and
\begin{align*}
\delta \bar C^1 &= (\cosh\beta) \delta C^1 - (\sinh\beta) \delta J^2\\
\delta \bar C^2 &= (\cosh\beta) \delta C^2 + (\sinh\beta) \delta J^1\\
\delta \bar C^3 &= \delta C^3. 
\end{align*}
The same applies to classical angular momentum and center of mass. The case of classical angular momentum corresponds to the combination of (3.27), (3.28), (3.30) in \cite{AKL}. 
\end{remark}
\subsection{Non-radiative spacetimes}
In this subsection, we prove that both classical and invariant charges for classical BMS fields are invariant under boost.
\begin{theorem}\label{boost_invariance_nonradiative}
Suppose $\bar Y$ and $Y$ are two classical BMS fields related by a boost. We have
\begin{align*}
\tilde Q(\bar u, \bar Y) = \tilde Q(u, Y) \quad\mbox{and}\quad Q(\bar u, \bar Y)= Q(u, Y)
\end{align*}
at any retarded time $\bar u$ and $u$.
\end{theorem}
\begin{proof}
We deal with the two cases $\bar Y = \bar Y^0 \partial_{\bar u}$ and $\bar Y = \bar u \bar Y^1 \partial_{\bar u}+ \bar Y^a \partial_a$ separately.\\ 
\noindent Case 1: Suppose $\bar Y = \bar Y^0 \partial_{\bar u}$. Note that $\tilde Q(u,Y) = Q(u,Y)$. For any retarded time $\bar u$ and $u$, we have \begin{align*}
 Q(\bar u, \bar Y) &= \int_{S^2} \bar Y^0(\bar x^a) (2\mathring{ \bar m}) d\bar\sigma \\
&=  \int_{S^2} K(\bar x) \bar Y^0(\bar x) (2\mathring m\circ g)(\bar x) d\hat\sigma \\
&=  \int_{S^2} K(g^{-1}(x)) \bar Y^0(g^{-1}(x)) \cdot 2\mathring m(x) d\sigma \\
&= Q(u,Y). 
\end{align*}
\noindent Case 2: $\bar Y = \bar u \bar Y^1 \partial_{\bar u}+ \bar Y^a \partial_a$. For any retarded time $\bar u$ and $u$, we have, by Lemma \ref{transform_boost},
\begin{align*}
\tilde Q(\bar u, \bar Y)&= \int_{S^2} \lt[ \bar Y^a ( \bar N_a - \frac{1}{4} \bar\sigma^{de}\bar C_{ae} \bar\na^b \bar C_{db} ) + \frac{1}{16} \bar\na_a \bar Y^a \bar\sigma^{ab} \bar\sigma^{de} \bar C_{ad} \bar C_{be} + \bar u \bar\na_a \bar Y^a \bar m \rt] \,d\bar\sigma\\
&= \int_{S^2} \bar Y^a K^2 \pl_a g^A \lt( N_A(u, g(\bar x)) + (K\bar u - u)(\na_A \mathring m - \frac{1}{4} \na^B \mathring P_{BA}) \rt) \,d\bar\sigma\\
&\quad + \int_{S^2} \bar Y^a \lt( 3\mathring m K^2 \na_a K \bar u - \frac{3}{4} K^2 \pl_a g^A \mathring P_{BA} \hat\na^b K \bar u \pl_b g^B \rt) d\bar\sigma\\
&\quad + \int_{S^2} \bar Y^a \lt( -\frac{1}{4} \rt)K^2 \hat\sigma^{de} K^{-1} \pl_a g^A \pl_e g^E \mathring C_{AE} \lt( -K^{-2} \bar\na^b K \cdot K \bar C_{db} + K^{-1} K^2 \pl_b g^B \na^D \mathring C_{BD} \rt)\, d\bar\sigma\\
&\quad + \int_{S^2} \frac{1}{16}  \lt( \hat\na_a \bar Y^a - 2K^{-1} \hat\na_a K \bar Y^a \rt)K^2 \mathring C_{BD} \mathring C^{BD} \,d\bar\sigma\\
&\quad + \int_{S^2} \bar u \lt( \hat\na_a \bar Y^a - 2K^{-1} \hat\na_a K \bar Y^a \rt) K^3 \mathring m \,d\bar\sigma
\end{align*}
where $\mathring m, \mathring P_{BA}, \mathring C_{AB}$ are evaluated at $g(\bar x)$.
We simplify the integral into
\begin{align*}
\tilde Q(\bar u, \bar Y) &= \int_{S^2} Y^A \lt( N_A(u,x) - u \na_A \mathring m + u \frac{1}{4} \na^B \mathring P_{BA} \rt) \,d\sigma\\
&\quad + \int_{S^2} Y^A \lt( -\frac{1}{4} \mathring C_{AB} \na_D \mathring C_{BD} \rt) + \frac{1}{16} \na_A Y^A \mathring C_{BD} \mathring C^{BD} \,d\sigma\\
&\quad + \int_{S^2} Y^A \lt( (K\circ g^{-1}) \bar u \na_A \mathring m + \na_A (K\circ g^{-1}) \bar u \mathring m \rt) + \bar u \na_A Y^A (K\circ g^{-1}) \mathring m \,d\sigma\\
&\quad + \int_{S^2} Y^A \bar u \lt( -\frac{1}{4} (K\circ g^{-1}) \na^B \mathring P_{BA} - \frac{3}{4} \mathring P_{BA} \na^B (K\circ g^{-1}) \rt) \,d\sigma.
\end{align*}
We have
$\na^B \mathring P_{BA} = \frac{1}{2} \epsilon_{AE}\na^E \Delta(\Delta + 2) \mathring \cb$ and integration by parts yields
\begin{align}\label{non-radiative_boost_bp}
\int_{S^2} Y^A \na^B \mathring P_{BA} \,d\sigma = \int_{S^2} -\frac{1}{2} (\Delta + 2) \lt( \epsilon_{AE} \na^E Y^A \rt) \cdot \Delta \mathring\cb =0
\end{align} for either $Y^A = \epsilon^{AB}\na_B \tilde X^k$ or $Y^A = \na^A \tilde X^k$. Hence the first two lines are precisely $\tilde Q(u, \bar Y) $. Moreover, the third line vanishes through integration by parts. Finally, since $\na_A\na_B(K \circ g^{-1}) - \frac{1}{2} \Delta (K\circ g^{-1})=0$ by \eqref{bar_traceless_Hessian}, the last line vanishes by Lemma \ref{integral formula for nonradiative case}. It follows that for any retarded time $\bar u$ and $u$,
\[ \tilde Q(\bar u,\bar Y) =\tilde Q(u, Y).\]
For the invariant charge, we have
\begin{align*}
Q(\bar u, \bar Y) &= \tilde Q(\bar u, \bar Y) + \int_{S^2} \lt( \bar Y^a \na_a \bar c - \bar Y^1 \bar c \rt)\mathring{\bar m} d\bar\sigma \\
&= \tilde Q(\bar u, Y) + \int_{S^2} \lt( \bar Y^a \na_a (K^{-1} \mathring c \circ g) - \bar Y^1\cdot K^{-1} \mathring c \circ g \rt) K \lt( \mathring m \circ g \rt) d\hat\sigma\\
&= \tilde Q(\bar u, Y) + \int_{S^2} \lt( Y^A \na_A \mathring c - Y^1 \mathring c \rt) \mathring m d\sigma\\
&= Q(\bar u, Y) = Q(u,Y)
\end{align*}
for any retarded time $\bar u$ and $u$. This completes the proof.
\end{proof}

\part{Extended BMS fields}

\section{Supertranslation invariance of the invariant charges}

In this section, we generallize the supertranslation invariance of invariant charge for classical BMS fields in Section \ref{sec:supertranslation_invariance_classical} to extended BMS fields. We first introduce the extended BMS algebra. The ensuing three subsections correspond exactly to those in Section \ref{sec:supertranslation_invariance_classical}.  

We recall from \cite[page 31]{Geroch} that a BMS field $Y$ (or an infinitesimal symmetry in the terminology of  \cite{Geroch} ) satisfies the following equations:
\[\mathcal{L}_Y \sigma=2\kappa \sigma, \mathcal{L}_Y n=-\kappa n\] for a function $\kappa$ on $\mathscr{I}^+$. Here $\sigma$ is the degenerate 2-metric and $n$ is the null generator of $\mathscr{I}^+$ which can be taken to be $\partial_u$ and $\mathcal{L}$ is the Lie derivative on $\mathscr{I}^+$. The solutions are of the form
\[ Y=(Y^0(x)+uY^1 (x))\partial_u+Y^A (x) \partial_A\] where $Y^A$ satisfies 
\begin{equation}\nabla^A Y^B+\nabla^B Y^A-2Y^1 \sigma^{AB}=0.\end{equation}

In other words, $Y^A$ is a vector field that satisfies the conformal Killing equation on $S^2$:
\begin{equation} \label{cke} \nabla^A Y^B+\nabla^B Y^A-\na_C Y^C \sigma^{AB}=0.\end{equation}
Note that \eqref{cke} is equivalent to 
\begin{equation}\label{ckf}\na^A Y^B=\frac{1}{2}(\na_CY^C)\sigma^{AB}+\frac{1}{2}(\epsilon_{CD}\na^C Y^D)\epsilon^{AB}.  \end{equation}

Smooth vector fields on $S^2$ that satisfy \eqref{cke} are linear combinations of of  $\nabla^A \tilde{X}^k, k=1, 2, 3$ and $\epsilon^{AB} \nabla_B \tilde{X}^k, k=1, 2, 3$. They are infinitesimal deformation of conformal diffeomorphisms or M\"obius transformations of $S^2$.

 An extended BMS field corresponds to the case when $Y^A$ is a solution of \eqref{cke} on an open subset $\mathscr{U} $ of $S^2$. We define
\begin{definition}\label{extended_conformal_Killing}
An {\it extended conformal Killing field} is a vector field $Y^A \frac{\pl}{\pl x^A}$ defined on an open set $\mathscr{U} \subset S^2$ such that $Y^A$ satisfies the conformal Killing equation \eqref{cke} on $\mathscr{U}$.
\end{definition}
In this case, we may not be able to integrate $Y^A$ to a conformal diffeomorphism of $S^2$. See the appendix for an explicit construction with $\mathscr{U} = S^2 - \{ (0,0,1)\}$. 
%and thus 
%\begin{equation}\label{ckf2}\na^A Y^B-\na^B Y^A=(\epsilon_{CD}\na^C Y^D) \epsilon^{AB}. \end{equation}

\begin{lemma} \label{eigen_form} An extended conformal Killing field $Y^A$ on $\mathscr{U} $ satisfies the following equations on $\mathscr{U} $:
\begin{align}
\nabla_A\nabla^AY^B+Y^B&=0 \label{eigen_form1}\\
 (\Delta +2)(\nabla_AY^A)&=0\label{eigen_fcn_div}\\
 (\Delta+2)(\epsilon_{AB}\na^A Y^B) &=0\label{eigen_fcn_curl}
\end{align}
\end{lemma}

\begin{proof}
To prove the first equality, we apply $\na_A$ to \eqref{cke} and obtain
\[\na_A (\na^A Y^B + \na^B Y^A) = \na^B \na_C Y^C.\] On the other hand, since $S^2$ has constant curvature $1$, we have
\[\na_A\na^B Y^A=\na^B \na_AY^A+Y^B.\]
Therefore,
\[ \na_A \na^A Y^B + \na^B \na_A Y^A + Y^B = \na^B \na_C Y^C,\] and the first equality follows.

For the second equality, we proceed as follows.
\[
\begin{split}
 \Delta (\nabla_AY^A) =& \nabla_B \nabla^B \nabla_AY^A\\
 =&\nabla_B \nabla_A \nabla^B Y^A -\nabla_B Y^B\\
 =& \frac{1}{2}\nabla_B\nabla^A \Big[(\nabla_CY^C)\sigma^{BA} +(\epsilon_{CD}\nabla^C Y^D)\epsilon^{BA}\Big]   -\nabla_B Y^B\\
 =& \frac{1}{2}\Delta (\nabla_AY^A)  -\nabla_B Y^B,
\end{split}
\] where we use \eqref{ckf} in the third equality above. 

The third equality can be checked similarly.
\end{proof}

\subsection{Supertranslation ambiguity of the classical charge}\label{subsec:total_flux_extended}
The derivation in this subsection is similar to Section \ref{subsec:supertranslation_ambiguity}. To ensure the finiteness and the absolute convergence of the total flux for an extended BMS field, we need to make stronger decay conditions. We also need to deal with extra terms that vanish for a classical BMS field. Finally, as mentioned after Definition 1.6, we need to assume that the data $m$, $N_A$, $C_{AB}$ and $N_{AB}$ are supported on $\mathscr{U} $. To preserve the support under a supertranslation, we assume that $f$ is also supported on $\mathscr{U} $.

Recalling \eqref{classical_evolution}, the evolution of classical charge of extended BMS fields reads
 \begin{equation}\label{charge_evolution_extended}\begin{split}
\partial_u \tilde Q(u, Y) &= \frac{1}{4}\int_{\mathscr{U} } Y^A \lt[ C_{AB}\nabla_{D}N^{BD} -N_{AB}\nabla_{D}C^{BD} +\frac{1}{2} \epsilon_{AB} \nabla^B (\epsilon^{PQ} C_P^{\,\,\, E}N_{EQ})        \rt]\\
&\quad -\frac{1}{4}\int_{\mathscr{U} } uY^1|N|^2 +\int_{S^2}  Y^0 \partial_u (2m)\\
&\quad +\frac{1}{4}\int_{\mathscr{U}} \lt[ u  (2Y^1)\na^A\na^BN_{AB} -  u (\epsilon_{AB} \nabla^A Y^B) \epsilon^{PQ} \nabla_P \nabla^E N_{EQ} \rt] \\
\end{split}\end{equation}

\begin{theorem}\label{extended_flux_ambiguity} Suppose $Y$ is an extended BMS field and $N_{AB}=O(|u|^{-2-\epsilon})$, then the supertranslation ambiguity of the total flux of the classical charge $\delta \tilde{Q} (Y)$ is 
\[\begin{split}
   &\delta \tilde Q(\bar Y) - \delta \tilde Q(Y) \\
= & \int_{\mathscr{U} }(-f \nabla_AY^A +2Y^A \nabla_A f )(m(+)- m(-))  - \frac{1}{4}\int_{\mathscr{U} } f  \lt[ \na^A\na^B(\na_D Y^D)     {C}_{AB}  \rt]_{-\infty}^{+\infty}\end{split}
\] if $\bar{Y}$ and $Y$ are related by the supertranslation $S_f$ \eqref{supertranslation_related}.
\end{theorem}

\begin{proof}

By Proposition \ref{evolution1}, we obtain 
\begin{equation}\begin{split}\label{total_flux1}
&\delta \tilde{Q}(Y)-\int_\mathscr{U}  2Y^0(m(+)-m(-))\\
=& \frac{1}{4}\int_{-\infty}^{+\infty}\int_{\mathscr{U} } Y^A \lt( C_{AB}\nabla_{D}N^{BD} -N_{AB}\nabla_{D}C^{BD} +\frac{1}{2} \epsilon_{AB} \nabla^B (\epsilon^{PQ} C_P^{\,\,\, E}N_{EQ}  )             \rt)\\
&-\frac{1}{4}\int_{-\infty}^{+\infty}\int_{\mathscr{U} } uY^1|N|^2 \\
& + \frac{1}{4}\int_{-\infty}^{\infty}\int_{\mathscr{U} }    u \lt[ \na^A\na^B(\na_CY^C)N_{AB} - \nabla^E\nabla_P (\epsilon_{AB} \nabla^A Y^B) \epsilon^{PQ}  N_{EQ} \rt]  \\
=& \RN{1}+\RN{2}+\RN{3}.
\end{split}\end{equation}where the last term was integrated by parts twice.

If $\bar{Y}$ is related to $Y$ by a supertranslation $S_f$ \eqref{supertranslation_related}, we have

\begin{equation}\label{after}\begin{split}
&\delta \tilde{Q}(\bar{Y})-\int_\mathscr{U}  2Y^0(m(+)-m(-))\\
=& \frac{1}{4}\int_{-\infty}^{+\infty}\int_{\mathscr{U} } Y^A \lt( \bar C_{AB}\nabla_{D}\bar{N}^{BD} -\bar{N}_{AB}\nabla_{D}\bar{C}^{BD} +\frac{1}{2} \epsilon_{AB} \nabla^B (\epsilon^{PQ} \bar{C}_P^{\,\,\, E}\bar{N}_{EQ} )\rt)\\
&-\frac{1}{4}\int_{-\infty}^{+\infty}\int_{\mathscr{U} } \bar{u}Y^1|\bar N|^2 \\
&+\frac{1}{4}\int_{-\infty}^{+\infty}\int_{\mathscr{U} }  \bar u  \lt[  \na^A\na^B(\na_DY^D)\bar N_{AB} -  \na_P \na^E  (\epsilon_{AB} \nabla^A Y^B) \epsilon^{PQ}\bar N_{EQ}  \rt] 
=\bar{\RN{1}}+\bar{\RN{2}}+\bar{\RN{3}}\\
\end{split}\end{equation}

To evaluate $\bar{\RN{1}}-\RN{1}$, we compute 
\begin{align*}
\nabla^B ( \bar{C}_P^{\,\,\, E}\bar{N}_{EQ} )& =     \partial_u (({C}_P^{\,\,\, E}-F_P^{\,\,\, E}){N}_{EQ}  \na^B f)+ \na_B( ({C}_P^{\,\,\, E}-F_P^{\,\,\, E}){N}_{EQ} ) \\
\bar{C}_{AB} \na_D \bar{N}^{BD}& = \partial_u ((C_{AB}-F_{AB}) N^{BD} \na_D f)-N_{AB}N^{BD}\nabla_Df + (C_{AB}-F_{AB})\na_D N^{BD} \\
\bar{N}_{AB} \na_D \bar{C}^{BD} &= N_{AB} N^{BD} \na_D f + N_{AB}(\na_D C^{BD}  -\na_D F^{BD}).\\
\end{align*}
All left hand sides are evaluated at $(\bar{u}, x)$ and all right hand sides are evaluated at $(\bar{u}+f, x)$. 
Plugging these into \eqref{after_classical}, applying Fubini's theorem and change of variable and the assumption \eqref{news_decay_1}, and using $2N_{AB}N^{BD}=|N|^2\delta_A^D$, we obtain, \[\begin{split}&\bar{\RN{1}}-\RN{1}\\
&=\frac{1}{4}\int_{-\infty}^{+\infty}\int_{\mathscr{U} }\lt[ Y^A \lt( -F_{AB}\nabla_{D}N^{BD} +N_{AB}\nabla_{D}F^{BD}-\frac{1}{2} \epsilon_{AB} \nabla^B (\epsilon^{PQ} F_P^{\,\,\, E}N_{EQ} )  -\na_A f |N|^2\rt) \rt]  \\
&=\frac{1}{4}\int_{-\infty}^{+\infty}\int_{\mathscr{U} }(-\nabla_AY^A f + 2 Y^A \nabla_A f )(\nabla_B\nabla_D N^{BD})-\frac{1}{4}\int_{-\infty}^{+\infty} \int_{\mathscr{U} }Y^A\na_A f |N|^2\\
&-\frac{1}{4} \int_{-\infty}^{+\infty}\int_{\mathscr{U} }\int f\epsilon^{PQ} N_{EQ} \na^E\na_P(\epsilon_{CD}\na^C Y^D)\\
\end{split}\]

On the other hand,
\[\begin{split}\bar{\RN{2}}-\RN{2}&=-\frac{1}{4}\int_{-\infty}^{+\infty} \int_{\mathscr{U} } (-f)(\frac{1}{2}\na_A Y^A)|{N}|^2)\\
&= \frac{1}{4}\int_{-\infty}^{+\infty} \int_{\mathscr{U} } \frac{1}{2}f(\na_A Y^A)|{N}|^2\end{split}\]

Therefore $\bar{\RN{1}}-\RN{1}+\bar{\RN{2}}-\RN{2}$ is given by

\begin{equation}\label{I+II}\begin{split}&\bar{\RN{1}}-\RN{1}+\bar{\RN{2}}-\RN{2}\\
&=\int_{-\infty}^{+\infty}\int_{\mathscr{U} }(-\nabla_AY^A f + 2 Y^A \nabla_A f )(\partial_u m)-\frac{1}{4} \int_{-\infty}^{+\infty}\int_\mathscr{U}  f\epsilon^{PQ} N_{EQ} \na^E\na_P(\epsilon_{CD}\na^C Y^D)\\
\end{split}\end{equation}

Next we deal with $\bar{\RN{3}}-\RN{3}$. We compute 

\begin{equation} \label{III}\begin{split}\bar{\RN{3}}-\RN{3}= \frac{1}{4}\int_{-\infty}^{\infty}\int_{\mathscr{U} }    f \lt[  \nabla^E\nabla_P (\epsilon_{AB} \nabla^A Y^B) \epsilon^{PQ}  N_{EQ}-\na^A\na^B(\na_CY^C)N_{AB}  \rt] 
      \end{split}  \end{equation}
using $\bar u = u -f$, \eqref{basic_data} and Fubini's theorem.

Adding up \eqref{I+II} and \eqref{III}, we obtain
\begin{align*}
4(\bar{\RN{1}}-\RN{1}+\bar{\RN{2}}-\RN{2}+\bar{\RN{3}}-\RN{3})
&= 4\int_{-\infty}^{+\infty}\int_{\mathscr{U} }(-f \nabla_AY^A +2Y^A \nabla_A f )\partial_u m\\
&\quad -\int_{-\infty}^{+\infty} \int_{\mathscr{U} }f \na^A\na^B(\na_C Y^C)     {N}_{AB}.
\end{align*}

Therefore, the supertranslation ambiguity is

\[\begin{split}&4(\delta \tilde{Q}(\bar{Y})-\delta \tilde{Q}({Y}))\\
=&4\int_{\mathscr{U} }(-f \nabla_AY^A +2Y^A \nabla_A f )m|_{-\infty}^{+\infty}-\int_{\mathscr{U} } f  \lt[ \na^A\na^B(\na_C Y^C)     {C}_{AB}  \rt]_{-\infty}^{+\infty}\\
\end{split}\]
\end{proof}

\subsection{Supertranslation invariance of the total flux of the invariant charge}
Finally, we prove that the supertranslation invariance of the invariant charge of an extended BMS field.
\begin{theorem} \label{extended_invariant} Suppose $Y$ is an extended BMS field and $N_{AB}=O(|u|^{-2-\epsilon})$, then the total flux of the  invariant charge $\delta Q(Y)$ is supertranslation invariant. Namely,
\[
\begin{split}
 \delta Q(\bar Y)- \delta Q(Y)=&\int_{\mathscr{U} }(-f_{l\le 1} \nabla_AY^A +2Y^A \nabla_A f_{l\le 1}  )(m(+)- m(-))\\
 & -\frac{1}{8} \lt[ \int_{\mathscr{U} } f_{\ell\le 1} \na^A\na^B (\na_C Y^C) C_{AB} \rt]_{-\infty}^{+\infty} \\
 =& -\delta Q ([ f_{\ell\leq 1} \partial_u, Y]) -\frac{1}{8} \lt[ \int_{\mathscr{U} } f_{\ell\le 1} \na^A\na^B (\na_C Y^C) C_{AB} \rt]_{-\infty}^{+\infty}
% &- \frac{1}{4} T(f_{l\le 1} ,c(+)-c(-))+ 
%\frac{1}{8} S(f_{l\le 1} ,\underline{c}(+)-
%\underline{c}(-)).
 \end{split}
\]
if $\bar{Y}$ is related to $Y$ by a supertranslation $S_f$ \eqref{supertranslation_related}. The last integral vanishes if the potentials of $C_{AB}(\infty)- C_{AB}(-\infty)$ is supported in $\mathscr{U}$.
\end{theorem}

\begin{proof}
We consider the following bilinear forms of functions on $S^2$
\[
T(v,w)=\int_{\mathscr{U} } \Big [ v(\na_A\na_B w -\frac{1}{2} \sigma_{AB} \Delta w)   \Big] \nabla^A\nabla^B (\nabla_CY^C)
\]
\[
S(v,w)= \int_{\mathscr{U} } \Big[v   (\epsilon_{AD}\nabla_D\nabla_B+\epsilon_{BD}\nabla_D\nabla_A) w \Big] \na^A\na^B(\na_C Y^C) 
\]
By Lemma \ref{T} in the appendix, $T(v,w)=T(w,v)$ and $S(v,w)=S(w,v)$ if  either $v$ or $w$ is supported in $\mathscr{U}$. In terms of the bilinear forms $T$ and $S$, 
\[
Q(Y)= {\tilde{Q}}(Y)+ \int_{\mathscr{U} } (Y^A \na_A c - Y^1 c)m- \frac{1}{16}T(c,c)- \frac{1}{32}S(c,\underline{c}) +\frac{1}{16} T(\underline{c},\underline{c}) - \frac{1}{32}S(\underline{c},c)
\]
and 
 $ \delta {\tilde{Q}}(\bar Y)- \delta {\tilde{Q}}(Y)$ is
\[
\int_{\mathscr{U} }(-f \nabla_AY^A +2Y^A \nabla_A f )(m(+)- m(-)) - \frac{1}{4} T(f,c(+)-c(-))- \frac{1}{8}S(f,\underline c(+)- \underline c(-)).
\]
Suppose $Y$ and $\bar Y$ are related by a supertranslation $S_f$, then we have
\[
\begin{split}
    & \delta Q(\bar Y)- \delta Q(Y)\\
= & \delta{\tilde{Q}}(\bar Y)- \delta {\tilde{Q}}(Y) + \lt[ \int_\mathscr{U}  (Y^A \na_A \bar c - Y^1 \bar c ) \bar m \rt]_{-\infty}^{\infty} - \lt[ \int_\mathscr{U}  (Y^A \na_A  c - Y^1 c ) m \rt]_{-\infty}^{\infty}\\
  &+ \frac{1}{16}  \Big [T(c(+),c(+)) -T(c(-),c(-))- T(\bar c(+), \bar c(+)) + T(\bar c(-), \bar c(-))    \Big]\\
  &+\frac{1}{32}    \Big [  S(c(+),\underline c(+))  -S(c(-),\underline c(- )) - S( \bar c(+), \underline{\bar c}(+)  )  +   S( \bar c(-), \underline{\bar c}(-) ))  \Big]  \\
  &+\frac{1}{32}    \Big [ S(\underline c(+),c(+) ) -S(\underline c(-).c(-)  ) -S( \underline{\bar c}(+) .\bar c(+) )  + S(  \underline{\bar c}(-), \bar c(-))  )  \Big] 
\end{split}
\]

After the supertranslation, we have $\bar c(\pm) = c(\pm) - 2f_{l\ge 2}$, $\underline {\bar{ c}}(\pm) = \underline c(\pm) $  and $\bar m(\pm) = m(\pm)$. As a result,
\begin{align*} 
   &\lt[ \int_\mathscr{U}  (Y^A \na_A \bar c - Y^1 \bar c ) \bar m \rt]_{-\infty}^{\infty} - \lt[ \int_\mathscr{U}  (Y^A \na_A  c - Y^1 c ) m \rt]_{-\infty}^{\infty}\\
=& -2 \int_{\mathscr{U} }  (Y^A \na_A f_{l \ge 2} -Y^1 f_{l \ge 2} ) (m(+)-m(-)). 
\end{align*}
On the other hand,
\[
\begin{split}
 &T(c(+),c(+)) -T(c(-),c(-))- T(\bar c(+), \bar c(+)) + T(\bar c(-), \bar c(-))  \\
=&2 T( c(+),  f_{l \ge 2})  +  2T(  f_{l \ge 2}, c(+) )  - 2T( c(-),  f_{l \ge 2})  - 2 T(  f_{l \ge 2}, c(-) )  \\
= &2  T(  f_{l \ge 2},c(+)-c(-))  + 2 T( c(+)-c(-),  f_{l \ge 2})\\ 
=&  2  T(  f_{l \ge 2},c(+)-c(-))  + 2 T( c(+)-c(-),  f)\\
=& 2T ( f_{l \ge 2} + f , c(+)-c(-)).
\end{split}
\]
Similarly, 
\[
\begin{split}
   &    S(c(+),\underline c(+))  -S(c(-),\underline c(- )) - S( \bar c(+), \underline{\bar c}(+)  )  +   S( \bar c(-), \underline{\bar c}(-) )  \\
  =& 2 S(f_{l \ge 2}, \underline c(+)- \underline c(-))
     \end{split}
\]
and
\[
\begin{split}
 &S(\underline c(+),c(+) ) -S(\underline c(-).c(-)  ) -S( \underline{\bar c}(+) .\bar c(+) )  + S(  \underline{\bar c}(-), \bar c(-))   \\
  =&2 S( f, \underline c(+)- \underline c(-)).
\end{split}
\]
It follows that 
\[
\begin{split}
 \delta Q(\bar Y)- \delta Q(Y)=&\int_{\mathscr{U} }(-f_{l\le 1} \nabla_AY^A +2Y^A \nabla_A f_{l\le 1}  )(m(+)- m(-))\\
 & -\frac{1}{8} \lt[ \int_{\mathscr{U} } f_{\ell\le 1} \na^A\na^B (\na_C Y^C) C_{AB} \rt]_{-\infty}^{+\infty} 
% &- \frac{1}{4} T(f_{l\le 1} ,c(+)-c(-))+ 
%\frac{1}{8} S(f_{l\le 1} ,\underline{c}(+)-
%\underline{c}(-)).
 \end{split}
\]
It remains to show that the last integral vanishes, given the additional assumption of the potentials of $C_{AB}(\infty)-C_{AB}(-\infty)$. Combining Lemma 2.6, (2.18), (2,7) of \cite{CKWWY} and \eqref{eigen_fcn_div}, we have
\begin{align*}
\int_{\mathscr{U} } \na^A \na^B (\na_C Y^C) (C_{AB}(\infty)-C_{AB}(-\infty)) = \frac{1}{2} \int_{\mathscr{U} } \Delta(\Delta + 2) (\na_C Y^C) \cdot (c(+)-c(-)) =0 
\end{align*}
and
\begin{align*}
&\int_{\mathscr{U} } \tilde X^k \na^A \na^B (\na_C Y^C)(C_{AB}(\infty)-C_{AB}(-\infty)) \\
&= \frac{1}{2}\int_{\mathscr{U} } \tilde X^k \lt[ (\Delta+2)(\na_C Y^C) (\Delta+2)(c(+)-c(-))  - 2 \epsilon^{AB} \na_A (\na_C Y^C) \na_B (\Delta + 2) (\underline{c}(+) - \underline{c}(-)) \rt]\\
&=0.
\end{align*}
This completes the proof using \eqref{supertranslation_bracket}.
\end{proof}

\subsection{Non-radiative spacetimes}
Recall that the invariant charge is independent of the retarded time in a non-radiative spacetime. We show that in non-radiative spacetimes the invariant charge itself, not just its total flux, is invariant under supertranslation. Namely, fixing $\bar{u}=\bar{u}_0$, we consider the invariant charges
\[\begin{split}&Q=Q(u_0,Y)\\
&\bar Q=Q(\bar{u}_0,\bar Y)
\end{split}\]
and prove the following theorem:

\begin{theorem} \label{invariant_CWY_vanish_news_extended}
The invariant charge satisfies
\[\begin{split}
\bar{Q}-Q= &\int_{\mathscr{U}} (2 Y^A \na_A f_{\ell\le 1} - f_{\ell\le 1} \na_A Y^A)  \mathring m  -\frac{1}{8}\int_{\mathscr{U}} f_{\ell\leq 1}  \na^A\na^B(\na_C Y^C)     \mathring{C}_{AB} \\
=& -Q ([ f_{\ell\leq 1} \partial_u, Y])  -\frac{1}{8}\int_{\mathscr{U}} f_{\ell\leq 1}  \na^A\na^B(\na_C Y^C)     \mathring{C}_{AB}
 \end{split}\]
 Moreover,
 \[\int_{\mathscr{U}} f_{\ell\leq 1}  \na^A\na^B(\na_C Y^C)     \mathring{C}_{AB}=0\]
 if the potentials for $ \mathring{C}_{AB}$ are supported in $\mathscr{U}$.
\end{theorem}
\begin{proof}

Taking the difference of $\bar Q $ and  $Q$ and applying \eqref{transform}, we obtain 
 \begin{equation}
 \begin{split}
    &\bar Q-Q\\
 =&\frac{1}{4}\int_{\mathscr{U}}  Y^A \Big [  \mathring{C}_{AB}\nabla_D F^{BD}+F_{AB} \nabla_D \mathring{C}^{BD}+\frac{1}{2} \nabla_A (\mathring{C}_{BD} F^{BD})- F_{AB} \nabla_D F^{BD} \\
      &\qquad\qquad\qquad -\frac{1}{4} \na_A( F_{BD}F^{BD} )   \Big ] +\int_{\mathscr{U}} Y^A \lt[\bar{N}_A  (\bar{u}_0, x)-  N_A  (u_0, x)\rt]\\
& + \big[ \int_{\mathscr{U}}(2Y^1 f_{\ell\ge 2} - 2Y^A \na_A f_{\ell\ge 2})  \mathring m\big]  + \frac{1}{16} T( \mathring{c}, \mathring{c}) - \frac{1}{16} T( \mathring{c} -2f_{\ell\geq 2}, \mathring{c}- 2f_{\ell\geq 2})\\
&\qquad\qquad\qquad +\frac{1}{32} S( \mathring{c},\mathring{\underline{c}}) -\frac{1}{32} S( \mathring{c}-2f_{\ell\geq 2},\mathring{\underline{c}}) +\frac{1}{32} S( \mathring{\underline{c}},\mathring{c}) -\frac{1}{32} S(\mathring{\underline{c}}, \mathring{c}-2f_{\ell\geq 2})\\
&+ (\bar u_0 - u_0) \Big[ \int_{\mathscr{U}} 2 \mathring m Y^1 -\frac{1}{4} \int_{\mathscr{U}}  (\epsilon_{AB} \nabla^A Y^B) \epsilon^{PQ} \nabla_P \nabla^E \mathring{C}_{EQ} \Big]
 \end{split}
 \end{equation}
 We have 
  \[
 \begin{split}
 &2 T( \mathring{c}, \mathring{c}) - 2T( \mathring{c} -2f_{\ell\geq 2}, \mathring{c}- 2f_{\ell\geq 2})+S( \mathring{c},\mathring{\underline{c}}) -( \mathring{c}-2f_{\ell\geq 2},\mathring{\underline{c}}) +S( \mathring{\underline{c}},\mathring{c}) -S(\mathring{\underline{c}}, \mathring{c}-2f_{\ell\geq 2})\\
=  &4\int_{\mathscr{U}}( f_{\ell\geq 2}+f)  \lt[ \na^A\na^B(\na_C Y^C)    \mathring {C}_{AB}  \rt]  + 8T(f_{\ell\geq 2} ,f_{\ell\geq 2} ).
 \end{split}
 \]
 and observe that 
  \[ \int_{\mathscr{U}} Y^A \lt[  -F_{AB}\nabla_D F^{DB}-\frac{1}{4} \na_A( F_{BD}F^{BD} ) \rt]  = T(f_{\ell\geq 2} ,f_{\ell\geq 2})\] 
  by Lemma \ref{quadratic_integral}.  
  
  We compute
\[\begin{split}&\int_{\mathscr{U}} Y^A \lt[\bar{N}_A  (\bar{u}_0, x)-  N_A  (u_0, x)\rt] + (\bar u_0 - u_0) \Big[ \int_{\mathscr{U}} 2  \mathring m Y^1 -\frac{1}{4} \int_{S^2}  (\epsilon_{AB} \nabla^A Y^B) \epsilon^{PQ} \nabla_P \nabla^E C_{EQ} \Big]
\\
=& \int_{\mathscr{U}} Y^A\lt[     f(\na_A \mathring{m}+3\mathring{m}\nabla_A f-f \nabla^B \mathring{P}_{BA} -3\mathring{P}_{BA} \na^B f  
 \rt]. 
\end{split}
\]
As a result,
\[\begin{split}
   &\bar{Q}-Q\\
=&\frac{1}{4}\int_{\mathscr{U}} Y^A \lt[  \mathring{C}_{AB}\nabla_D F^{BD}+F_{AB} \nabla_D \mathring{C}^{BD}+\frac{1}{2} \nabla_A (\mathring{C}_{BD} F^{BD})-f \nabla^B \mathring{P}_{BA} -3\mathring{P}_{BA} \na^B f   \rt]\\
&+ \int_{\mathscr{U}}  (2 Y^A \na_A f_{\ell\le 1} - f_{\ell\le 1} \na_A Y^A) \mathring m+\frac{1}{8}\int_{\mathscr{U}}( f_{\ell\geq 2} +f) \na^A\na^B(\na_C Y^C)     \mathring{C}_{AB}.  \end{split}\]
By Lemma \ref{integral formula for nonradiative case}, the first line is equal to
\[ -\frac{1}{4} \int_{\mathscr{U}} f \na^A\na^B (\na_C Y^C) \mathring C_{AB}.\]

Putting everything together, we arrive at 
\[\begin{split}
\bar{Q}-Q= \int_{\mathscr{U}}  (2 Y^A \na_A f_{\ell\le 1} - f_{\ell\le 1} \na_A Y^A)  \mathring m  -\frac{1}{8}\int_{\mathscr{U}} f_{\ell\leq 1}  \na^A\na^B(\na_C Y^C)     \mathring{C}_{AB}.  \end{split}\]
This completes the proof using \eqref{supertranslation_bracket}.
\end{proof}

\section{Boost invariance of the classical and invariant charges}
This section generalizes the boost invariance of classical and invariant charges in Section \ref{sec:boost_classical} to extended BMS fields. The two subsections correspond exactly to those of Section \ref{sec:boost_classical}.

We also recall the assumption made in the beginning of Section \ref{subsec:total_flux_extended} that $m,N_A,C_{AB},N_{AB}$ are supported in $\mathscr{U}$, the domain of the extended conformal Killing vector field $Y^A$.
\subsection{Total flux of classical and invariant charges}
We generalize Theorem \ref{boost_invariance_flux} to extended BMS fields. We begin with a pointwise identity.
\begin{lemma}\label{Y1_conformal}
\begin{align*}
\hat\na_a \hat\na_b (K\bar Y^1) - \frac{1}{2} \hat\Delta(K\bar Y^1) \hat\sigma_{ab}  = K \lt( \hat\na_a \hat\na_b (Y^1\circ g) - \frac{1}{2} \hat\Delta (Y^1 \circ g) \hat\sigma_{ab} \rt).
\end{align*}
\end{lemma}
\begin{proof}
Recalling
\[ \bar Y^1 = Y^1 \circ g - K^{-1} \pl_a K \bar Y^a\]
and \[ \hat\na_a \hat\na_b K - \frac{1}{2} \hat\Delta K \hat\sigma_{ab} =0,\]
we get
\begin{align*}
&\hat\na_a \hat\na_b (K\bar Y^1) - \frac{1}{2} \hat\Delta(K\bar Y^1) \hat\sigma_{ab} \\
&= K \lt(\hat\na_a \hat\na_b (Y^1\circ g - K^{-1}\pl_c K \bar Y^c) - \frac{1}{2} \hat\Delta(Y^1 \circ g - K^{-1}\pl_c K \bar Y^c) \hat\sigma_{ab} \rt) \\
&\quad + \pl_a K \pl_b \bar Y^1 + \pl_b K \pl_a \bar Y^1 - \hat\na^c K \pl_c \bar Y^1 \hat\sigma_{ab}\\
&= K \lt( \hat\na_a \hat\na_b (Y^1\circ g) - \frac{1}{2} \hat\Delta (Y^1 \circ g) \hat\sigma_{ab} \rt) - K \hat\na_a \hat\na_b \lt( K^{-1}\pl_c K \bar Y^c \rt)\\
&\quad + \pl_a K \pl_b \bar Y^1 + \pl_b K \pl_a \bar Y^1 + \mathscr{A}  
\end{align*}
where $\mathscr{A}$ is proportional to $\hat\sigma_{ab}$.
We compute
\begin{align*}
&-\hat\na_a\hat\na_b \lt( K^{-1}\pl_c K \bar Y^c \rt) \\
&= \hat\na_a \lt( K^{-2}\pl_b K \pl_c K \bar Y^c - K^{-1} \cdot \frac{1}{2}\hat\Delta K \hat\sigma_{bc}\bar Y^c - K^{-1} \pl_c K \hat\na_b \bar Y^c \rt)\\
&= -2K^{-3} \pl_a K \pl_b K \pl_c K \bar Y^c + K^{-2}\cdot \frac{1}{2}\hat\Delta K \hat\sigma_{ab} \pl_c K \bar Y^c + K^{-2} \pl_b K \cdot \frac{1}{2}\hat\Delta K \hat\sigma_{bc}\bar Y^c +K^{-2} \pl_b K \pl_c K \hat\na_a \bar Y^c\\
&\quad +K^{-2} \pl_a K \cdot \frac{1}{2} \hat\Delta K \hat\sigma_{bc} \bar Y^c - K^{-1} \frac{1}{2}\pl_a \hat\Delta K \hat\sigma_{bc} \bar Y^c - K^{-1}\cdot \frac{1}{2}\hat\Delta K \hat\sigma_{bc}\bar Y^c\\
&\quad + K^{-2}\pl_a K \pl_c K \hat\na_b \bar Y^c - K^{-1}\cdot \frac{1}{2}\hat\Delta K \hat\sigma_{ac} \bar\na_b \bar Y^c - K^{-1}\pl_c K \hat\na_a \hat\na_b \bar Y^c.
\end{align*}
We use the fact that a conformal Killing vector field is preserved by a conformal map
\begin{align}\label{conf_Killing_hat}
\hat\na^b \bar Y^c + \hat\na^c \bar Y^b = 2 (Y^1\circ g) \hat\sigma^{bc}
\end{align} to simplify the last term
\begin{align*}
\hat\na_a \hat\na_b \bar Y^c &= \hat\na_a \lt( -\hat\sigma_{bd} \hat\na^c \bar Y^d + 2 (Y^1\circ g) \delta_b^c \rt)\\
&= -\hat\sigma_{bd} \lt( \hat\na^c\hat\na_a \bar Y^d + \hat R_{a\;\;\;\;e}^{\;\;cd} \bar Y^e \rt) + 2 \pl_a \lt( \bar Y^1 + K^{-1}\pl_d K \bar Y^d \rt)\delta_b^c\\
&= -\hat\sigma_{bd} \hat\na^c\hat\na_a \bar Y^d + (\delta_b^c \hat\sigma_{ae}- \hat\sigma_{ba} \delta^c_e ) \bar Y^e + 2 \pl_a \bar Y^1 \delta_b^c\\
&\quad +2 \lt( -  K^{-2} \pl_a K \pl_d K \bar Y^d + K^{-1}\hat\Delta K \hat\sigma_{ad}\bar Y^d + K^{-1}\pl_d K \hat\na_a \bar Y^d \rt)\delta_b^c 
\end{align*}
Noting that $\hat\Delta K = -2 K + const.$, we obtain 
\begin{align*}
-\hat\na_a\hat\na_b \lt( K^{-1}\pl_c K \bar Y^c \rt) &= -\frac{1}{2} \lt( \hat\na_a\hat\na_b \lt( K^{-1}\pl_c K \bar Y^c \rt) + \hat\na_b\hat\na_a \lt( K^{-1}\pl_c K \bar Y^c \rt)\rt) \\
&= \frac{1}{2}K^{-1} \pl_c K \hat\na^c (\hat\sigma_{bd} \hat\na_a \bar Y^d + \hat\sigma_{ad}\hat\na_b \bar Y^d)  - K^{-1} \lt( \pl_b K \pl_a \bar Y^1 + \pl_a K \pl_b \bar Y^a \rt).
\end{align*}
By \eqref{conf_Killing_hat}, the first term on the right-hand side is proportional to $\hat\sigma_{ab}$. Consequently, the difference of two symmetric traceless 2-tensors $\hat\na_a \hat\na_b (K\bar Y^1) - \frac{1}{2} \hat\Delta(K\bar Y^1) \hat\sigma_{ab} $ and $K \lt( \hat\na_a \hat\na_b (Y^1\circ g) - \frac{1}{2} \hat\Delta (Y^1 \circ g) \hat\sigma_{ab} \rt)$ is proportional to $\hat\sigma_{ab}$ and hence zero. This completes the proof.
\end{proof}

\begin{theorem}\label{boost_equ_extended_flux}
Suppose $\bar Y$ and $Y$ are two extended BMS fields related by a boost $B_g$ where $\bar Y$ (resp. $Y$) is defined on $(-\infty,+\infty) \times \mathscr{U} $ (resp. $(-\infty,+\infty) \times g(\mathscr{U})$). If $N_{AB} = O(|u|^{-2-\epsilon})$, then we have
\[ \delta \tilde Q(\bar Y) = \delta \tilde Q(Y)\quad\mbox{and}\quad \delta Q(\bar Y) = \delta Q(Y).\]
\end{theorem}
\begin{proof}
Following the proof of Theorem \ref{boost_invariance_flux}, we have
\begin{align*}
\delta\tilde Q(\bar Y) - \delta\tilde Q(Y) &= -\frac{1}{4} \int_{-\infty}^{+\infty} \int_\mathscr{U} \bar u \lt( \bar\epsilon_{ab} \bar\na^a \bar Y^b \rt)\bar\epsilon^{pq}\bar\na_p \bar\na^e \bar N_{eq} d\bar\sigma d\bar u \\
&\quad + \frac{1}{4} \int_{-\infty}^{+\infty} \int_{g(\mathscr{U})} u (\epsilon_{AB}\na^A Y^B) \epsilon^{PQ} \na_P\na^E N_{EQ} d\sigma du
\end{align*} for extended BMS fields. By \eqref{closed=coclosed}, \eqref{Y1_conformal}, and the change of variable formula, we have 
\begin{align*}
&\int_{-\infty}^{+\infty} \int_\mathscr{U} \bar u \lt( \bar\epsilon_{ab} \bar\na^a \bar Y^b \rt)\bar\epsilon^{pq}\bar\na_p \bar\na^e \bar N_{eq} d\bar\sigma d\bar u\\
&=\int_{-\infty}^{+\infty} \int_{\mathscr{U}} \bar u (\bna^a \bna^b \bar Y^1 - \frac{1}{2}\bar\Delta \bar Y^1 \bar\sigma^{ab})\bar N_{ab} d\bar\sigma d\bar u \\
&= \int_{-\infty}^{+\infty} \int_{\mathscr{U}} u K^{-1} \lt( \hat\na_a \hat\na_b (K\bar Y^1) - \frac{1}{2} \hat\Delta (K\bar Y^1) \hat\sigma_{ab} \rt) \hat\sigma^{ac} \hat\sigma^{bd} \bar N_{cd} d\hat\sigma du\\
&= \int_{-\infty}^{+\infty} \int_{g(\mathscr{U})} u (\na^A \na^B Y^1 - \frac{1}{2}\Delta Y^1 \sigma^{AB}) N_{AB} d\sigma du\\
&= \int_{-\infty}^{+\infty} \int_{g(\mathscr{U})} u (\epsilon_{AB}\na^A Y^B) \epsilon^{PQ} \na_P\na^E N_{EQ} d\sigma du.
\end{align*}
This proves $\delta \tilde Q(\bar Y)=\delta \tilde Q(Y)$ for two extended BMS fields $\bar Y$ and $Y$ related by a boost.

For the invariant charge, it suffices to examine the correction terms
\begin{align}\label{temp_3}
-\frac{1}{16} \int_{g(\mathscr{U})} c \na^A\na^B (\na_D Y^D) C_{AB} \Big|_{-\infty}^{\infty} - \frac{1}{16} \int_{g(\mathscr{U})} \cb \na^A\na^B (\na_D Y^D) (*C)_{AB} \Big|_{-\infty}^{\infty}.
\end{align}

We claim that each of them is invariant under boost. By Lemma \ref{Y1_conformal}, we have the relation
\begin{align*}
\int_{\mathscr{U}} \bar c \bna^a\bna^b \bar Y^1 \bar C_{ab} d\bar\sigma &= \int_{\mathscr{U}} \bar c \lt( \hat\na_a \hat\na_b (K\bar Y^1) - \frac{1}{2} \hat\Delta(K\bar Y^1) \hat\sigma_{ab} \rt) \hat\sigma^{ad}\hat\sigma^{be} (K\bar C_{de}) d\hat\sigma \\
&= \int_{\mathscr{U}} K \bar c \lt( \hat\na_a \hat\na_b (Y^1 \circ g) - \frac{1}{2}\hat\Delta (Y^1 \circ g) \hat\sigma_{ab} \rt)\hat\sigma^{ad}\hat\sigma^{be} (K\bar C_{de}) d\hat\sigma \\
&= \int_{g(\mathscr{U})} c \lt( \na^A \na^B Y^1 - \frac{1}{2}\Delta Y^1 \sigma^{AB} \rt)C_{AB} d\sigma.
\end{align*}
at $\bar u = +\infty$ and $u = +\infty$. Hence, the first term in \eqref{temp_3} is invariant. The invariance of the second term follows similarly. This completes the proof of $\delta Q(\bar Y) = \delta Q(Y)$ for two extended BMS fields related by a boost.  
\end{proof}

\subsection{Non-radiative spacetimes}
We generalize Theorem \ref{boost_invariance_nonradiative} to extended BMS fields.
\begin{theorem}\label{boost_equ_extended_vanish}
Suppose $\bar Y$ and $Y$ are two extended BMS fields related by a boost $B_g$ in a non-radiative spacetime where $\bar Y$ (resp. $Y$) is defined on $(-\infty,+\infty) \times \mathscr{U} $ (resp. $(-\infty,+\infty) \times g(\mathscr{U})$). If $\mathring\cb$ is supported in $g(\mathscr{U})$, then 
\[ \tilde Q(\bar u, \bar Y) = \tilde Q(u,Y) \quad\mbox{and}\quad Q(\bar u,\bar Y) = Q(u,Y)\]
for any retarded time $\bar u$ and $u$.\end{theorem}

\begin{remark}
The assumption on the support of $\mathring \cb$ is invariant under BMS transformations in non-radiative spacetimes.
\end{remark}

\begin{proof}
The proof of Theorem \ref{boost_invariance_nonradiative} implies that  
\begin{align*}
\tilde Q(\bar u, \bar Y) &= \tilde Q(\bar u, Y) -\frac{1}{4} \int_{g(\mathscr{U})} \bar u Y^A \lt( (K\circ g^{-1}) \na^B \mathring P_{BA} + 3 \mathring P_{BA} \na^B (K\circ g^{-1}) \rt) \\
&\quad + \frac{1}{4} \int_{g(\mathscr{U})} \bar u (K \circ g^{-1}) \na^A \na^B (\na_C Y^C) \mathring C_{AB}
\end{align*}
for extended BMS fields. Note that the assumption on $\mathring\cb$ is used so that the integration by parts in \eqref{non-radiative_boost_bp} holds. By Lemma \ref{integral formula for nonradiative case}, the additional terms vanish. This proves $\tilde Q(\bar u,\bar Y) = \tilde Q(u,Y)$ for any retarded time $\bar u$ and $u$. 

For the invariant charge, one verifies the boost invariance of each correction term as in the previous theorem. This completes the proof of $Q(\bar u,\bar Y) = Q(u,Y)$ for any retarded time $\bar u$ and $u$.  
\end{proof}

\appendix

\section{Integral formulae for extended conformal Killing fields}
We derive integral lemmas for extended gconformal Killing fields $Y^A \frac{\pl}{\pl x^A}$ defined on an open subset $\mathscr{U} \subset S^2$, see Definition \ref{extended_conformal_Killing}. 
In these integral formulae, we impose assumptions so that integrations by parts are valid. This happens, for example, if the extended conformal Killing field is integrating against tensors with compact support in $\mathscr{U}$. If $Y^A \frac{\pl}{\pl x^A}$ is a global conformal Killing field, namely $\mathscr{U}=S^2$, those assumptions hold automatically and the formulae are applicable in Part 1.

\begin{lemma} \label{integral_lemma} Let $Y^A \frac{\pl}{\pl x^A}$ be an extended conformal Killing field defined on an open subset $\mathscr{U} \subset S^2$. If $f$ is a function with compact support in $\mathscr{U}$, then we have
\begin{align}\label{integral_formula}
\begin{split}
&\int_\mathscr{U} Y^A \lt( -F_{AB}\nabla_{D}N^{BD} +N_{AB}\nabla_{D}F^{BD}-\frac{1}{2} \epsilon_{AB} \nabla^B (\epsilon^{PQ} F_P^{\,\,\, E}N_{EQ} )  \rt)\\
&= \int_\mathscr{U} (-\nabla_AY^A f + 2 Y^A \nabla_A f )(\nabla_B\nabla_D N^{BD}) -\int_\mathscr{U} f\epsilon^{PQ} N_{EQ} \na^E\na_P(\epsilon_{CD}\na^C Y^D)
\end{split}
\end{align}
where $F_{AB}=2\nabla_A \nabla_B f-\Delta f\sigma_{AB}$.
\end{lemma}
\begin{proof}
We denote the function $\alpha = \epsilon_{AB}\na^A Y^B$ in the proof. We write the left-hand side as $(1)+(2)+(3)$. Integrating by parts, we get
\begin{align*}
(1) &= \int_\mathscr{U} -Y^A (2\na_A\na_B f -\Delta f\sigma_{AB}) \na_D N^{BD} \\
&= \int_\mathscr{U} 2 \na^B Y^A \na_A f \na^D N_{BD} + 2 Y^A \na_A f \na^B\na^D N_{BD} + \Delta f Y^A \na^B N_{AB},\\
(2) &= \int_\mathscr{U} Y^A N_{AB}\na^B (\Delta+2)f\\
&= \int_\mathscr{U} -\na^B Y^A N_{AB}(\Delta+2)f -Y^A \na^B N_{AB} (\Delta+2)f\\
&= \int_\mathscr{U} -Y^A \na^B N_{AB}(\Delta+2)f
\end{align*}
and hence
\begin{align*}
(1)+(2) = \int_\mathscr{U} 2\na^B Y^A \na_A f \na^D N_{BD} + 2 Y^A \na_A f \na^B\na^D N_{BD} - 2Y^A f \na^B N_{AB}.
\end{align*}
We simplify the first integral
\begin{align*}
&\int_\mathscr{U} 2\na^B Y^A \na_A f \na^D N_{BD} \\
&= \int_\mathscr{U} \na_C Y^C \na^B f \na^D N_{BD} + \alpha \epsilon^{BA} \na_A f \na^D N_{BD}\\
&= \int_\mathscr{U} -\na_B\na_C Y^C f \na_D N^{BD} -\na_C Y^C f \na^B\na^D N_{BD} + \alpha \epsilon^{BA} \na_A f \na^D N_{BD}\\
&= \int_\mathscr{U} 2Y^B f \na^D N_{BD} -\epsilon_{BC}\na^C \alpha f \na_D N^{BD} -\na_C Y^C f \na^B\na^D N_{BD} + \alpha \epsilon^{BA} \na_A f \na^D N_{BD}.
\end{align*}
In the last equality, we used the identity
\begin{align}\label{Hodge Y}
\na_B \na_C Y^C = -2Y_B +\epsilon_{BC}\na^C \alpha,
\end{align}
which is obtained from the computation 
\[ \na_B \na_C Y^C = \na_C \na_B Y^C - Y_B = \na^C \lt( \frac{1}{2} \na_D Y^D \sigma_{BD} + \frac{1}{2}\alpha \epsilon_{BC} \rt) - Y_B.\]
We also note the identity
\begin{align}\label{closed=coclosed}
2 \na_A\na_B (\na_C Y^C) - \Delta (\na_C Y^C)\sigma_{AB} = \epsilon_{AC}\na_B \na^C (\epsilon_{DE}\na^D Y^E) + \epsilon_{BC}\na_A \na^C (\epsilon_{DE}\na^D Y^E),
\end{align}
which is obtained by differentiating \eqref{Hodge Y} and using the identity $(\Delta+2)\na_C Y^C =0$.

In summary, we have
\[ (1)+(2) = \int_\mathscr{U} (-\nabla_AY^A f + 2 Y^A \nabla_A f )(\nabla_B\nabla_D N^{BD}) -\epsilon_{BC}\na^C \alpha f \na_D N^{BD}+ \alpha \epsilon^{BA} \na_A f \na^D N_{BD}.\]
Finally, we have
\begin{align*}
(3) &= \int_\mathscr{U} -\frac{1}{2}\alpha \epsilon^{PQ} (2 \na_P \na^E f - \Delta f \delta_P^E) N_{EQ}\\
&= \int_\mathscr{U} \na_P \alpha\epsilon^{PQ} \na^E f N_{EQ} + \alpha\epsilon^{PQ}\na^E f \na_P N_{EQ}\\
&= \int_\mathscr{U} -\na^E \na_P \alpha \epsilon^{PQ} f N_{EQ} -\na_P \alpha  \epsilon^{PQ} f\na^E N_{EQ} + \alpha \na^E f \epsilon_{EP} \na_Q N^{EQ} 
\end{align*}
where we used the identity $\epsilon^{PQ}\na_P N_{EQ} = \epsilon^{EP}\na_Q N^{EQ}$ \cite[(2.13)]{CKWWY} in the last equality. Putting these together, the assertion follows.
\end{proof}

\begin{lemma} \label{T}  Let $\alpha$ be a function defined on an open subset $\mathscr{U} \subset S^2$. Consider two functions $u$ and $v$ defined on $\mathscr{U}$ such that one of them has compact support. Denoting $u_{AB} = \na_A\na_B u$ and $v_{AB} = \na_A\na_B v$, we have
  \begin{equation}\label{T_symmetry} \int_\mathscr{U} \Big [ u(v_{AB} -\frac{1}{2} \Delta v \sigma_{AB}) - v(u_{AB} -\frac{1}{2} \Delta u \sigma_{AB})  \Big] \nabla^A\nabla^B \alpha=\frac{1}{2}\int_\mathscr{U} (u\Delta v-v\Delta u) (\Delta+2) \alpha.\end{equation}
and
   \begin{align}\label{S_symmetry}\begin{split} \int_\mathscr{U} \Big [ u({\epsilon_A}^D v_{DB}+{\epsilon_B}^D v_{DA} ) - v({\epsilon_A}^D u_{DB}+{\epsilon_B}^D u_{DA} )  \Big] \nabla^A\nabla^B \alpha \\
   =\int_\mathscr{U} \epsilon_{AD}\lt( u \na^A v - v \na^A u\rt) \na^D (\Delta+2) \alpha.\end{split} \end{align}

In particular, applying the above formulae to the divergence of an extended conformal Killing field $Y^A \frac{\pl}{\pl x^A}$ defined on $\mathscr{U}$, we have
\[\int_\mathscr{U} \Big [ u(v_{AB} -\frac{1}{2} \Delta v \sigma_{AB}) - v(u_{AB} -\frac{1}{2} \Delta u \sigma_{AB})  \Big] \nabla^A\nabla^B (\nabla_CY^C)=0 .\]
\[ \int_\mathscr{U} \Big [ u({\epsilon_A}^D v_{DB}+{\epsilon_B}^D v_{DA} ) - v({\epsilon_A}^D u_{DB}+{\epsilon_B}^D u_{DA} )  \Big] \nabla^A\nabla^B (\nabla_CY^C)=0\]
if either $u$ or $v$ is supported in $\mathscr{U}$. 
\end{lemma}
\begin{proof}
We write $u_A = \na_A u$ in the proof. We compute
\[
\begin{split}
    & \int_\mathscr{U}(uv_{AB}  - vu_{AB} ) \nabla^A\nabla^B \alpha \\
 = & \int_\mathscr{U} \nabla_B (uv_{A}  - vu_{A} )   \nabla^A\nabla^B \alpha\\
 = & -\int_\mathscr{U}  (uv_{A}  - vu_{A} )\nabla_B\nabla^A\nabla^B \alpha\\
 =&-\int_\mathscr{U}  (uv_{A}  - vu_{A} )\big[ \nabla^A\nabla_B\nabla^B(\nabla_CY^C)   + \nabla^A (\nabla_CY^C)    \Big]\\
 =&-\int_\mathscr{U}  (uv_{A}  - vu_{A} ) \nabla^A(\Delta+1)\alpha  \\
 =&\int_\mathscr{U}  ( u \Delta v - v \Delta u)(\Delta+1)\alpha\\
 \end{split}
 \]
\eqref{T_symmetry} follows from rearranging the terms. For \eqref{S_symmetry}, note that we can replace $\na^A\na^B \alpha$ by the symmetric traceless 2-tensor $\na^A\na^B \alpha - \frac{1}{2}\Delta\alpha \sigma^{AB}$. Recalling Proposition 2.4 of \cite{CKWWY} that $\epsilon_{DB} \na^D C^{BA} = \epsilon^{AD}\na^B C_{BD}$ for any symmetric traceless 2-tensor, we integrate by parts to get
\begin{align*}
&\int_\mathscr{U}  u ( {\epsilon_A}^D v_{DB}+{\epsilon_B}^D v_{DA} ) \nabla^A\nabla^B \alpha \\
&= \int_\mathscr{U} -\lt( u_D {\epsilon_A}^D v_B + u_A {\epsilon_B}^D v_D \rt)\lt( \na^A\na^B \alpha - \frac{1}{2}\Delta\alpha\sigma^{AB} \rt)\\
&\quad + \int_\mathscr{U} -u v_B \lt( -\frac{1}{2}{\epsilon_B}^D \na_D (\Delta+2 )\alpha \rt) -u{\epsilon_B}^D v_D \cdot \frac{1}{2}\na^B (\Delta+2)\alpha.
\end{align*}
Interchanging $u$ and $v$ and subtraction yield \eqref{S_symmetry}. The second claims follows from $(\Delta+2)(\nabla_CY^C)=0$ for an extended BMS field.
\end{proof}
The following lemma generalizes Lemma 2.3 of \cite{CKWWY}.
\begin{lemma} \label{quadratic_integral} Let $Y^A \frac{\pl}{\pl x^A}$ be an extended conformal Killing field defined on an open subset $\mathscr{U} \subset S^2$. If $u$ is a function with compact support in $\mathscr{U}$, then we have
\[
\begin{split}
    &\int_\mathscr{U} Y^1 [2\nabla_A\nabla_Bu \nabla^A\nabla^B u- (\Delta u)^2] .\\
=   &\int_\mathscr{U} Y^1[(\Delta+2)u]^2 + 2  \int_\mathscr{U}  u (\nabla^A\nabla^B Y^1 ) (\nabla_A\nabla_Bu - \frac{1}{2} \Delta u \sigma_{AB})
\end{split}
\]
where $2Y^1=\na_CY^C$
\end{lemma}
\begin{proof}
We use the following formulae in the derivation

\[\begin{split}
\Delta|\nabla u|^2&=2|\nabla^2u|^2+2\nabla u\cdot \nabla (\Delta+1)u\\
\Delta (u^2)&=2|\nabla u|^2+2 u\Delta u\\
\Delta(u\Delta u)&=(\Delta u)^2+2\nabla u\cdot \nabla(\Delta u)+u \Delta^2 u.
\end{split}\]

Integrating by parts twice gives

\[\int_\mathscr{U} Y^1\nabla_A\nabla_Bu \nabla^A\nabla^B u=\int_\mathscr{U}  u \nabla^A\nabla^B(Y^1 \nabla_A\nabla_Bu) \]

We compute 
\[\begin{split}
      &\nabla^A\nabla^B(Y^1 \nabla_A\nabla_Bu)\\
=&(\nabla^A\nabla^B Y^1) \nabla_A\nabla_Bu+2\nabla^B Y^1 \nabla^A \nabla_A\nabla_Bu+Y^1 \nabla^A\nabla^B \nabla_A\nabla_Bu\\
=&-Y^1\Delta u+2\nabla^BY^1\nabla_B (\Delta+1)u+Y^1\Delta (\Delta+1)u + (\nabla^A\nabla^B Y^1- \frac{1}{2}\Delta Y^1 \sigma_{AB} ) \nabla_A\nabla_Bu \\
=&Y^1 \Delta^2 u+2\nabla^B Y^1 \nabla_B (\Delta+1)u + (\nabla^A\nabla^B Y^1- \frac{1}{2}\Delta Y^1 \sigma_{AB} ) \nabla_A\nabla_Bu  \end{split},\] where we use $\nabla^A \nabla_A\nabla_Bu=\nabla_B (\Delta+1)u$ and $\Delta Y^1 = -2 Y^1$ in the second equality. 

On the other hand, we have the identity:

\[ 2\nabla^B u\nabla_Bv=\Delta (uv)- u\Delta v-v\Delta u\] and thus 
\[2\nabla^BY^1\nabla_B (\Delta+1)u=\Delta (Y^1 (\Delta+1)u)-Y^1 \Delta(\Delta+1)u+2Y^1  (\Delta+1)u.\]
Putting all together gives:
\[\begin{split} &\int_\mathscr{U} Y^1 \nabla_A\nabla_Bu \nabla^A\nabla^B u\\
=&\int_\mathscr{U} Y^1\Delta^2 u+\int_\mathscr{U} u  [\Delta (Y^1  (\Delta+1)u)-Y^1 \Delta(\Delta+1)u+2Y^1  (\Delta+1)u]\\
  & + \int_\mathscr{U}  u (\nabla^A\nabla^B Y^1- \frac{1}{2}\Delta Y^1 \sigma_{AB} ) \nabla_A\nabla_Bu\\
=&\int_\mathscr{U} Y^1 [(\Delta u)^2+2u\Delta u+2u^2] + \int_\mathscr{U}  u (\nabla^A\nabla^B Y^1- \frac{1}{2}\Delta Y^1 \sigma_{AB} ) \nabla_A\nabla_Bu\end{split}.\]

Therefore, 
\[
\begin{split}
    &\int_\mathscr{U} Y^1 [2\nabla_A\nabla_Bu \nabla^A\nabla^B u- (\Delta u)^2] .\\
=   &\int_\mathscr{U} Y^1[(\Delta+2)u]^2 + 2  \int_\mathscr{U}  u (\nabla^A\nabla^B Y^1- \frac{1}{2}\Delta Y^1 \sigma_{AB} ) \nabla_A\nabla_Bu\\
=   &\int_\mathscr{U} Y^1[(\Delta+2)u]^2 + 2  \int_\mathscr{U}  u (\nabla^A\nabla^B Y^1 ) (\nabla_A\nabla_Bu - \frac{1}{2}\Delta u \sigma_{AB})
\end{split}
\]

\end{proof}

The next lemma, used in all invariance proofs in non-radiative spacetimes, generalizes Lemma B.1 and B.2 in \cite{CKWWY}.
\begin{lemma}\label{integral formula for nonradiative case} Let $Y^A \frac{\pl}{\pl x^A}$ be an extended conformal Killing vector field defined on an open subset $\mathscr{U} \subset S^2$. Let $f$ be a function on $S^2$ and $C_{AB}$ be a symmetric traceless 2-tensor on $S^2$. $P_{BA} = \na_B\na^E C_{EA} - \na_A\na^E  C_{EB}$. If either $f$ or $C_{AB}$ has compact support in $\mathscr{U}$, then   
\begin{align*}
&\frac{1}{4}\int_\mathscr{U} Y^A \lt[  C_{AB}\nabla_D F^{BD}+F_{AB} \nabla_D C^{BD}+\frac{1}{2} \nabla_A ( C_{BD} F^{BD})-f \nabla^B P_{BA} -3P_{BA} \na^B f   \rt]\\
&= -\frac{1}{4} \int_\mathscr{U} f \na^A\na^B(\na_C Y^C) C_{AB}
\end{align*}
where $F_{AB} = 2\na_A \na_B f - \Delta f \sigma_{AB}$ and $P_{BA} = \na_B \na^E C_{EA} - \na_A \na^E C_{EB}$.
\end{lemma}
\begin{remark}
If $Y^A \frac{\pl}{\pl x^A}$ is defined globally on $S^2$, then we have
\[
\na^A\na^B(\na_C Y^C) C_{AB} = 0.
\]
\end{remark}
\begin{proof}
Integrating by part the last two terms of the first line, we have
\[
\begin{split}
   &\int_{\mathscr{U}}Y^A(-f \nabla^B P_{BA} -3P_{BA} \na^B f )\\
=& \int_{\mathscr{U}}-2 Y^A (\na_B\na^D C_{DA} - \na_A\na^D C_{DB})\na^B f +2f (\na^B Y^A  - \frac{1}{2} \na_EY^E \sigma^{AB})\na_B\na^D C_{DA} \\
=& \int_{\mathscr{U}} 2 \na_B Y^A\na^D C_{DA}\na^B f + 2 Y^A\na^D C_{DA} \Delta f  - 2 \na_A Y^A\na^D C_{DB}\na^B f-2 Y^A\na^D C_{DB} \na_A\na^B f\\
           &  +\int_{\mathscr{U}}  2fY^A\na^D C_{DA}-2 \na^B Y^A  \na_B f\na^D C_{DA}-f \na_EY^E\na^A\na^B C_{AB} \\
=&  \int_{\mathscr{U}} - 2 Y^A \na^D C_{DB} (\na_A\na^B f - \frac{1}{2}\Delta f \sigma_{AB}) +  Y^A \na^D C_{AD} (\Delta + 2)f \\
   &  -\int_{\mathscr{U}}2 \na_A Y^A\na^D C_{DB}\na^B f+f \na_EY^E\na^A\na^B C_{AB}, 
\end{split}
\]
where in the first equality we anti-symmetrize $\na^B Y^A$ into $\na^B Y^A - \frac{1}{2} \na_E Y^E \sigma^{AB}$ by \eqref{cke}.

After simplifying the last two terms
\[
\begin{split}
  &\int_{\mathscr{U}}2 \na_A Y^A\na^D C_{DB}\na^B f+f \na_EY^E\na^A\na^B C_{AB}\\
=&\int_{\mathscr{U}}-2 \na^D\na_A Y^A C_{DB}\na^B f  -2\na_AY^A C_{DB} \na^D\na^B f \\
   &+ \int_{\mathscr{U}} \na^A \na^B f \na_EY^E C_{AB} + f  \na^A\na^B\na_EY^E C_{AB} + 2 \na^Bf \na^A (\na_EY^E) C_{AB} \\
 =& \int_{\mathscr{U}} f  \na^A\na^B\na_EY^E C_{AB} - \na_AY^A C_{DB} \na^D\na^B f\\
  =& \int_{\mathscr{U}} f  \na^A\na^B\na_EY^E C_{AB} - \frac{1}{2}\na_AY^A C_{DB} F^{DB},
\end{split}
\]
we get
\begin{align*}
&\frac{1}{4}\int_{\mathscr{U}} Y^A \lt[  C_{AB}\nabla_D F^{BD}+F_{AB} \nabla_D C^{BD}+\frac{1}{2} \nabla_A (C_{BD} F^{BD})-f \nabla^B P_{BA} -3P_{BA} \na^B f   \rt]\\
&= -\frac{1}{4} \int_{\mathscr{U}} f \na^A\na^B(\na_C Y^C) C_{AB}.
\end{align*}
\end{proof}

\section{Explicit forms of extended BMS fields}
An extended BMS field corresponds to a singular solution of \eqref{cke} which cannot be integrated to a conformal diffeomorphism of $S^2$. In the following, we discuss 
these solutions. First observe that the equation \eqref{cke} is conformally invariant in the sense that $Y^A$ is a solution of \eqref{cke} with respect to $\sigma$ if and on if 
$Y^A$ is a solution of \eqref{cke} with respect to a metric that is conformal to $\sigma$. 
Consider the stereographic projection $\rho$ from the complement of the north pole $(0, 0,1)$  of ${S}^2\subset \R^3$ to $\mathbb{R}^2=\mathbb{C}$:
\[\rho: S^2-\{(0, 0, 1)\}\rightarrow \mathbb{R}^2=\mathbb{C}.\] The pull-back of  the flat metric $\mathring{\sigma}=|dz|^2= (dx)^2+(dy)^2$ on $\mathbb{R}^2=\mathbb{C}$ through $\rho$ is conformal to the standard round metric $\sigma$ on $S^2$. With respect to the flat metric $\bar{\sigma}=\rho^*\mathring{\sigma}$, \eqref{cke}
is exactly the Cauchy-Riemann equation, i.e $Y^1\partial_x+Y^2\partial_y$ satisfies \eqref{cke} if and only if $Y^1+iY^2$ satisfies the Cauchy-Riemann equation $\partial_{\bar{z}}(Y^1+iY^2)=0$. 
A complete basis for the extended conformal Killing fields can be found in \cite{Flanagan-Nichols} in terms of $\ell=1$ spherical harmonics. 
\[ l_m=-z^{m+1} \partial_z, \bar{l}_m=-\bar{z}^{m+1} \partial_{\bar{z}}, m \in \mathbb{Z}.\]

In particular, $\partial_z$ and $z^{m+1}\partial_z$ are conformal Killing fields with respect to $\bar{\sigma}$.
We also recall that the function $z$ on $\mathbb{C}$ satisfies $\bar{\Delta}z=0$ and $\bar{\nabla}_A z\bar{\na}^Az=0$ with respect to $\bar{\sigma}$.

Let $Z$ be the pull-back of $z$ through $\rho$, and $\mathring{Y}^A$ be the pull back of $(\partial_z)^A$ through $\rho^{-1}$. $Z$ and $\mathring{Y}^A$ satisfy
\[\Delta Z=0, \na_AZ\na^A Z=0, \nabla^A \mathring{Y}^B+\nabla^B \mathring{Y}^A-\nabla_C\mathring{Y}^C \sigma^{AB}=0\] by the conformal invariance of these equations. 

Explicitly $Z$, $\bar{Z}$, and $\mathring{Y}^A$ are given by 
\[ Z=\frac{\tilde{X}^1+i\tilde{X}^2}{1-\tilde{X}^3}, \bar{Z}=\frac{\tilde{X}^1-i\tilde{X}^2}{1-\tilde{X}^3}\] and 
\[\begin{split}\mathring{Y}^A&=(\epsilon^{CD} \nabla_C Z\nabla_D \bar{Z})^{-1}\epsilon^{AB} \nabla_B \bar{Z} \end{split}.\] 
All are defined and smooth outside the north pole $(0, 0, 1)$. A straightforward calculation shows that
\begin{equation}\label{divergence}  \na_A \mathring{Y}^A=-(\tilde{X}^1-i\tilde{X}^2) \text{ and } \na_A Z\na^A\tilde{X}^3=Z.\end{equation}

\begin{lemma} The divergence $\na_A Y^A$ of the conformal Killing field $Y^A=Z^{m+1} \mathring{Y}^A$ satisfies 
\[ (\Delta+2)\na_A Y^A=0\] on $S^2-\{(0, 0, 1)\}$. 
\end{lemma}

\begin{proof} We compute \[\begin{split}\na_A Y^A&=\na_A (Z^{m+1} \mathring{Y}^A)\\
&=(m+1) Z^m (\na_A Z) \mathring{Y}^A+Z^{m+1} \na_A\mathring{Y}^A\\
&=(m+1) Z^m-Z^{m+1} (\tilde{X}^1-i\tilde{X}^2)\\
&=Z^m(m-\tilde{X}^3),\end{split}\] where we use the first equation in \eqref{divergence}. 

Furthermore, 
\[\begin{split}\Delta (\na_A Y^A)&=(\Delta Z^m)(m-\tilde{X}^3)-2m Z^{m-1}\na_A Z \na^A \tilde{X}^3-Z^m \Delta \tilde{X}^3\\
&=-2 Z^m(m-\tilde{X}^3),
\end{split}\] where we use the second equation in \eqref{divergence}.

\end{proof}
\begin{comment}
We should be able to use $\tilde{X}^1=\sin\theta\sin\phi, \tilde{X}^2=\sin\theta \cos\phi, \tilde{X}^3=\cos\theta$ to write $Z^m(m-\tilde{X}^3)=\cot^m(\frac{\theta}{2})(\sin\phi+i\cos\phi)^m(m-\cos\theta)$ and check directly with the form of the metric $d\theta^2+\sin^2\theta d\phi^2$. 
\end{comment}
\section{Conformal transformation on $S^2$}
Let $\sigma$ be the induced metric of unit sphere $S^2 \subset \R^3$ Let $g: S^2 \rw S^2$ be a conformal map and $\hat\sigma = g^*\sigma$. It is well-known that $g$ is a linear fractional transformation and a direct computation yields that $\hat\sigma = K^2 \bar \sigma$. In addition, both $\hat\sigma$ and $\bar \sigma$ have constant curvature 1. Moreover, 
\begin{align*}
K(x) = \frac{1}{\alpha_0 + \alpha_i x^i} >0
\end{align*} 
where $(\alpha_0,\alpha_i)$ is a unit timelike vector and $x^i,i=1,2,3$ forms an orthonormal basis of first eigenfunctions on $(S^2,\bar\sigma)$. 

We denote by $\bar\nabla,\bar\Delta$ ($\hat\nabla, \hat\Delta$ respectively) the covariant derivative and Laplacian with respect to $\bar\sigma$ ($\hat\sigma,$ respectively). Taking Hessian and Laplacian of $K^{-1}$, we get
\begin{lemma}
\begin{align}
K^{-2} \bar\nabla_a \bar\nabla_b K - 2K^{-3} \bar\nabla_a K \bar\nabla_b K &= \alpha_i x^i \bar\sigma_{ab} \label{Hessian K} \\
\frac{1}{2} K^{-2} \bar\Delta K - K^{-3}|\bar\nabla K|^2 = \alpha_i x^i \label{Laplace K}
\end{align}
\end{lemma}

The next lemma compares the covariant derivatives $\bar\nabla$ and $\hat\na.$ The proof is through direct computation and is left to the readers. 

 \begin{lemma}
\begin{enumerate}
\item  The Christoffel symbols of $\hat\sigma$ and $\bar\sigma$ are related by
 \[\hat{\Gamma}_{ab}^c=\bar\Gamma_{ab}^c + \frac{1}{2} K^{-2}\left( (\partial_b K^2) \delta_a^c+(\partial_a K^2) \delta_b^c-(\partial_d K^2) \bar\sigma_{ab} \bar\sigma^{cd}\right).\]
\item  The covariant derivatives of a covector $X_a$ are related by
 \[\begin{split}\hat\nabla_b X_a&=\bar\nabla_b X_a + K^{-1} (\partial_b K)X_a + K^{-1} (\partial_aK) X_b - K^{-1}(\na^cK)X_c \,\bar\sigma_{ab}\\
K^2 \hat{\sigma}^{ab} \hat{\nabla}_b X_a&={\bar\sigma}^{ab} {\bar\nabla}_b X_a.\end{split}\] In particular,  for a function $f$, we have 
\begin{align}\label{bar_Laplacian}
K^2\hat{\Delta} f=\bar\Delta f
\end{align} and 
\begin{align}
\hat\nabla_a \hat\nabla_b (Kf) - \frac{1}{2} \hat{\Delta} (Kf )\hat{\sigma}_{ab}   =K\lt( \bar\nabla_a \bar\nabla_b f -\frac{1}{2} \bar\Delta f \bar\sigma_{ab}\rt) \label{bar_traceless_Hessian}
\end{align}
\item If $C_{ab}$ is a traceless symmetric 2-tensor for $\bar\sigma$ (hence also traceless for $\hat\sigma$), then
\begin{align}\label{bar_div_traceless}
\hat\nabla^b (KC_{ab}) = K^{-1}\bar\nabla^b C_{ab} + K^{-2} \bar\nabla^b K C_{ab} = K^{-2} \bar\nabla^b (KC_{ab})
\end{align}
and
\begin{align}\label{bar_divdiv_traceless}
\hat\nabla^a \hat\nabla^b (KC_{ab}) = K^{-3} \bar\nabla^a \bar\nabla^b C_{ab}.
\end{align} 
\end{enumerate}
\end{lemma}

\end{document}